\documentclass[12pt]{iopart}

\usepackage{iopams}
\usepackage{graphicx}
\usepackage{pseudocode}
 \expandafter\let\csname equation*\endcsname\relax
  \expandafter\let\csname endequation*\endcsname\relax
\usepackage{amsmath}
\usepackage{amssymb}
\usepackage{amsthm}  
\usepackage{color}

\newcommand{\vm}[1]{\boldsymbol{#1}}
\renewcommand{\thepseudocode}{1}
\newtheorem{theorem}{Theorem}
\begin{document}

\title{Bayesian signal reconstruction for 1-bit compressed sensing}
\author{Yingying Xu$^{1}$, Yoshiyuki Kabashima$^{1}$ and Lenka Zdeborov\'{a}$^{2}$}

\address{$^{1}$Department of Computational Intelligence and Systems Science, 
\\Tokyo Institute of Technology, Yokohama 226-8502, Japan
\\$^{2}$Institut de Physique Th\'{e}orique, IPhT, CEA Saclay, and URA 2306, CNRS, 
\\F-91191 Gif-sur-Yvette, France}
\ead{yingxu@sp.dis.titech.ac.jp, kaba@dis.titech.ac.jp, {lenka.zdeborova@cea.fr}}

\begin{abstract}
The 1-bit compressed sensing framework enables the recovery of  a sparse vector $\vm{x}$ 
from the sign information of each entry of its linear transformation.
Discarding the amplitude information can significantly reduce the amount of data, which is highly beneficial in practical applications.
In this paper, we present a Bayesian approach to  signal reconstruction for 1-bit compressed sensing,  
and analyze its typical performance using statistical mechanics.
{As a basic setup,}
 {we consider the case that the measuring matrix $\vm{\Phi}$ has i.i.d entries, and the measurements $\vm{y}$ are noiseless.}
Utilizing the replica method, we show that the Bayesian approach 
enables better reconstruction than the $l_1$-norm minimization approach,  
asymptotically saturating the performance obtained when the non-zero entries positions of the signal are known, 
 {for signals whose non-zero entries follow zero mean Gaussian distributions.}
We also test a message passing algorithm for  signal reconstruction on the basis of 
belief propagation.
The results of numerical experiments are consistent with those of the theoretical analysis.
\end{abstract}

\maketitle

\section{Introduction}
\textit{
Compressed (or compressive) sensing} (CS) is currently one of the most popular topics in information science, 
and has been used for applications in various engineering fields, such as audio and visual electronics, medical imaging devices, and astronomical observations \cite{CShardware}.
Typically, smooth signals, such as natural images and communications signals, 
can be represented by a sparsity-inducing basis, such as a Fourier or wavelet basis \cite{Elad2010,StarckMurtaghFadili2010}. 
The goal of CS is to reconstruct a high-dimensional signal from its lower-dimensional linear transformation data, 
utilizing the prior knowledge on the sparsity of the signal \cite{CandesWakin2008}.
This results in time, cost, and precision advantages.

Mathematically, the CS problem can  be expressed as follows: an $N$-dimensional vector $\textrm{\boldmath $x^{0}$}$
is linearly transformed into an $M$-dimensional vector $\textrm{\boldmath $y$}$
by an $M \times N$-dimensional measurement matrix $\textrm{\boldmath $\Phi$}$, as $\textrm{\boldmath $y$}=\textrm{\boldmath $\Phi x^{0}$}$ \cite{CandesWakin2008}. 
The observer is free to choose the measurement protocol. Given $\textrm{\boldmath $\Phi$}$ and $\textrm{\boldmath $y$}$, 
the central problem is how to reconstruct $\textrm{\boldmath $x^{0}$}$.
When $M < N$, {due to the} loss of information, the inverse problem has infinitely many solutions. However, when it is guaranteed that $\textrm{\boldmath $x^{0}$}$ has only $K<M$ nonzero entries in some convenient basis (i.e., when the signal is sparse enough) and the measurement matrix is {incoherent} with that basis, there is a high probability that the inverse problem has a unique and exact solution. 
Considerable efforts have been made to clarify the condition for the uniqueness and correctness of the solution,  
and {to develop} practically feasible signal reconstruction algorithms \cite{Donoho2006, CandesRombergTao2006,KWT2009,Sompolinsky2010,Krzakala2012}. 

Recently, a scheme called {\em 1-bit compressed sensing (1-bit CS)} was proposed. In 1-bit CS, the signal is recovered from only the sign data of the linear measurements 
$\textrm{\boldmath $y$}=\mathrm{sign}\left(\textrm{\boldmath $\Phi x^{0}$}\right)$, where $\mathrm{sign}(x)=x/|x|$ for $x \ne 0$ is a component-wise operation when $x$ is a vector \cite{1bitCS}. 
Discarding the amplitude information can significantly reduce the amount of data to be stored and/or transmitted. 
This is highly advantageous for most real-world applications, particularly those in which the measurement is accompanied by the transmission of digital information \cite{Lee2012}.
In 1-bit CS, the amplitude information is lost during the measurement stage, making perfect recovery of the original signal impossible. 
Thus, we generally need more measurements to compensate for the loss of information.
The scheme is considered to have practical relevance in situations where perfect recovery is not required, and measurements are inexpensive but precise quantization is expensive.
These features are very different from those of general CS.

The most widely used signal reconstruction scheme in CS is $l_1$-norm minimization, which searches for the vector  with the smallest $l_1$-norm $||\textrm{\boldmath $x$}||_1=\sum_{i=1}^N |x_i|$ under the constraint $\vm{y}=\vm{\Phi x}$. 
This is based on the work of Cand{\`{e}}s et al.\ \cite{CandesWakin2008}--\cite{CandesRombergTao2006}, who also suggested the use of a random measurement 
matrix $\vm{\Phi}$ with independent and identically distributed 
entries. 
Because the optimization problem is convex and can be solved using efficient linear programming techniques, 
these ideas have led to various fast and efficient algorithms. The $l_1$-reconstruction is now widely used, and is 
responsible for the surge {of} interest in CS over the past few years.
Against this background, $l_1$-reconstruction was the first technique attempted in the development of the 1-bit CS problem. 
In \cite{1bitCS}, an approximate signal recovery algorithm was proposed based on the minimization of the $l_1$-norm 
under the constraint $\mathrm{sign}\left(\textrm{\boldmath $\Phi x$}\right)=\textrm{\boldmath $y$}$, and its utility was demonstrated by numerical experiments.
In \cite{YingKaba1bitCS2013}, the capabilities of this method were analyzed, and a new algorithm based on the cavity method was presented.
However, the significance of the $l_1$-based scheme may be rather weak for 1-bit CS, because 
the loss of convexity prevents  
the development of mathematically {guaranteed} and practically feasible algorithms.

Therefore, we propose 
another approach based on Bayesian inference for 1-bit CS, 
 {focused on the case that each entry of $\vm{\Phi}$ is independently generated from a standard Gaussian distribution, and the output $\vm{y}$ is noisless}.
Although the Bayesian approach is guaranteed to achieve {the} optimal performance when the actual signal distribution is given, 
quantifying the performance gain is a nontrivial task. 
We accomplish this task {utilizing} the replica method, 
which shows that 
 {when non-zero entries of the signal follow zero mean Gaussian distributions,}
the Bayesian optimal inference 
{asymptotically 
saturates the mean squared error (MSE) performance obtained when the positions of non-zero signal entries 
are known
as $\alpha=M/N \to \infty$.} 
This means that, 
 {in such cases,} 
at least in terms of MSEs, the correct prior knowledge of the sparsity asymptotically becomes 
as informative as the knowledge of the exact positions of the non-zero entries.  
Unfortunately, performing the exact Bayesian inference is computationally difficult. 
This difficulty is resolved by employing the generalized approximate message passing technique, 
which is regarded as a variation of belief propagation or the cavity method {\cite{GAMP, kabaUdaLPTN}}. 


This paper is organized as follows. The next section sets up the 1-bit CS problem. 
In section 3, we  examine the signal recovery performance achieved by the Bayesian scheme {utilizing} the replica method. 
In section 4, an approximate signal recovery algorithm based on belief propagation is developed.
The utility of this algorithm is tested and its asymptotic performance is analyzed in section 5. The final section summarizes our work.

\section{Problem setup and Bayesian optimality}
Let us suppose that entry $x_i^0$ $(i=1,2,\ldots,N)$ of an $N$-dimensional signal (vector) $\vm{x}^0 =(x_i^0)\in \mathbb{R}^N$ is independently generated from an identical sparse distribution:
\begin{equation}
P\left(x\right)=\left(1-\rho\right) \delta \left( x \right)
+ \rho \tilde{P} \left( x \right), 
\label{sparse}
\end{equation}
where $\rho \in [0,1]$ represents the density of nonzero entries in the signal, and $\tilde{P} (x)$ is a distribution function of $x \in \mathbb{R}$ 
that has a finite second moment and does not have finite mass at $x=0$. 
In 1-bit CS, the measurement is performed as 
\begin{equation}
\textrm{\boldmath $y$}=\mathrm{sign} \left( \textrm{\boldmath $\Phi x^{0}$} \right), 
\label{measurement}
\end{equation}
where $\mathrm{sign}(x )=x/|x|$ operates in a component-wise manner, and   
for simplicity we assume that each entry of the $M \times N$ measurement matrix $\vm{\Phi}$ 
is provided as a sample of a Gaussian distribution of zero mean and variance $N^{-1}$. 

We shall adopt the Bayesian approach to reconstruct the signal from the 1-bit measurement $\textrm{\boldmath $y$}$
assuming that $\vm{\Phi}$ is correctly known in the recovery stage. 
Let us denote an arbitrary recovery scheme for the measurement $\vm{y}$ as
$\hat{\vm{x}}(\vm{y})$, where we impose a normalization constraint $|\hat{\vm{x}}(\vm{y})|^2=N\rho$
to compensate for the loss of amplitude information by the 1-bit measurement.
Equations (\ref{sparse}) and (\ref{measurement}) indicate that, for a given $\vm{\Phi}$, 
the joint distribution of the sparse vector and its 1-bit measurement is 
\begin{eqnarray}
P(\textrm{\boldmath $x$}, \textrm{\boldmath $y$}|\vm{\Phi})=\prod_{\mu=1}^M \Theta\left ( y_\mu  (\vm{\Phi} \vm{x})_\mu \right )
\times \prod_{i=1}^N \left (\left(1-\rho\right) \delta \left( x_i \right)
+ \rho \tilde{P} \left( x_i \right) \right ), 
\label{joint}
\end{eqnarray}
where $\Theta(x)=1$ for $x > 0$, and vanishes otherwise. 
This generally provides $\hat{\vm{x}}(\cdot)$ with 
the mean square error, which is hereafter handled as the performance measure
for the signal reconstruction\footnote{ {Errors of other types, such as 
$l_p$-norm, can also be chosen as the performance measure. The argument shown in 
this section holds similarly even when such measures are used.}},  as follows:
\begin{eqnarray}
{\rm MSE}(\hat{\vm{x}}(\cdot))=
\sum_{\vm{y}} \int d \vm{x} P(\textrm{\boldmath $x$}, \textrm{\boldmath $y$}|\vm{\Phi})
\left |\frac{\hat{\vm{x}}(\vm{y})}{|\hat{\vm{x}}(\vm{y})|}-\frac{\vm{x}}{|\vm{x}|}\right |^2. 
\label{MSEdef}
\end{eqnarray}
The following theorem forms the basis of our Bayesian approach. 
\begin{theorem}
\label{theorem1}
${\rm MSE}(\hat{\vm{x}}(\cdot))$ is lower bounded as
\begin{eqnarray}
{\rm MSE}(\hat{\vm{x}}(\cdot)) \ge 2 \sum_{\vm{y}} P(\vm{y}|\vm{\Phi}) 
\left (1- \left |\left \langle \frac{\vm{x}}{|\vm{x}|} \right \rangle_{|\vm{y},\vm{\Phi}} \right | \right ), 
\label{MSE theorem}
\end{eqnarray}
where  
\begin{eqnarray}
P(\vm{y}|\vm{\Phi}) &=& \int d\vm{x} P(\vm{x}, \vm{y}|\vm{\Phi}) \cr
&=& \int d\vm{x} \prod_{\mu=1}^M \Theta\left ( y_\mu  (\vm{\Phi} \vm{x})_\mu \right )
\times \prod_{i=1}^N \left (\left(1-\rho\right) \delta \left( x_i \right)
+ \rho \tilde{P} \left( x_i \right) \right )
\label{marginal_y}
\end{eqnarray}
is the marginal distribution of the 1-bit measurement $\vm{y}$ and  
$\left \langle f(\vm{x}) \right \rangle_{|\vm{y},\vm{\Phi}}=\int d\vm{x} f(\vm{x}) P(\vm{x}|\vm{y},\vm{\Phi})
=\int d\vm{x} f(\vm{x}) P(\vm{x},\vm{y}|\vm{\Phi})/P(\vm{y}|\vm{\Phi})$ generally denotes the posterior mean of 
an arbitrary function of $\vm{x}$, $f(\vm{x})$, given $\vm{y}$.
The equality holds for the Bayesian optimal signal reconstruction 
\begin{eqnarray}
\hat{\vm{x} }^{\rm Bayes}(\vm{y})=\sqrt{N\rho}\left |\left \langle \frac{\vm{x}}{|\vm{x}|} \right \rangle_{|\vm{y},\vm{\Phi}} \right |^{-1}
\left \langle \frac{\vm{x}}{|\vm{x}|} \right \rangle_{|\vm{y},\vm{\Phi}}.
\label{Bayesian_reconstruction}
\end{eqnarray}
\end{theorem}
\begin{proof}
Employing the Bayes formula $P(\vm{x},\vm{y}|\vm{\Phi})=P(\vm{x}|\vm{y},\vm{\Phi})P(\vm{y}|\vm{\Phi})$ in
(\ref{MSEdef}) yields the expression 
\begin{eqnarray}
{\rm MSE}(\hat{\vm{x}}(\cdot))
&=&\sum_{\vm{y}} \int d \vm{x} P(\textrm{\boldmath $x$}, \textrm{\boldmath $y$} {|\vm{\Phi}})
\left |\frac{\hat{\vm{x}}(\vm{y})}{|\hat{\vm{x}}(\vm{y})|}-\frac{\vm{x}}{|\vm{x}|}\right |^2\cr
&=&\sum_{\vm{y}} \int d \vm{x} P(\textrm{\boldmath $x$}| \textrm{\boldmath $y$},\vm{\Phi})P(\textrm{\boldmath $y$}|\vm{\Phi})
\left( \left| \frac{\hat{\vm{x}}(\vm{y})}{|\hat{\vm{x}}(\vm{y})|}\right |^2 + \left | \frac{\vm{x}}{|\vm{x}|}\right |^2 - 2 \frac{\hat{\vm{x}}(\vm{y})\cdot\vm{x}}{|\hat{\vm{x}}(\vm{y})||\vm{x}|}\right) \cr
&=& 2\sum_{\vm{y}}P(\textrm{\boldmath $y$}|\vm{\Phi})\left(1-\frac{\hat{\vm{x}}(\vm{y})}{|\hat{\vm{x}}(\vm{y})|} \cdot \left \langle \frac{\vm{x}}{|\vm{x}|} \right \rangle_{|\vm{y},\vm{\Phi}} \right).
\label{MSE2}
\end{eqnarray}
Inserting the Cauchy--Schwarz inequality
\begin{eqnarray}
{\hat{\vm{x}}(\vm{y}) \cdot \left \langle \frac{\vm{x}}{|\vm{x}|} \right \rangle_{|\vm{y},\vm{\Phi}}}
\le 
{\left |\hat{\vm{x}}(\vm{y})\right | \left | \left \langle \frac{\vm{x}}{|\vm{x}|} \right \rangle_{|\vm{y},\vm{\Phi}} \right |} 
\end{eqnarray}
into the right-hand side of (\ref{MSE2}) yields the lower bound of (\ref{MSE theorem}), where 
the equality holds when $\hat{\vm{x}}(\vm{y}) $ is parallel to $\left \langle \frac{\vm{x}}{|\vm{x}|} \right \rangle_{|\vm{y},\vm{\Phi}}$. 
This, in conjunction with the normalization constraint of $\hat{\vm{x}}(\vm{y}) $, leads to 
(\ref{Bayesian_reconstruction}). 
\end{proof}

The above theorem guarantees that the Bayesian approach achieves the best possible performance 
in terms of MSE. Therefore, we hereafter focus on the reconstruction scheme of (\ref{Bayesian_reconstruction}), 
quantitatively evaluate its performance, and develop a computationally feasible approximate algorithm.

\section{Performance assessment by the replica method}
In statistical mechanics,  the macroscopic behavior of the system is generally analyzed by evaluating the partition function
or its negative logarithm, free energy. 
In our signal reconstruction problem, 
the marginal likelihood $P(\vm{y}|\vm{\Phi})$ of (\ref{marginal_y}) plays the role of the partition function. 
However, this still depends on the quenched random variables $\vm{y}$ and $\vm{\Phi}$. 
Therefore, we must further average the free energy as  
$\bar{f} \equiv -N^{-1} \left [\log P(\vm{y}|\vm{\Phi}) \right]_{\vm{y},\vm{\Phi}}$ 
to evaluate the typical performance, where
$\left [\cdots \right]_{\vm{y},\vm{\Phi}}$ denotes the configurational average 
concerning $\textrm{\boldmath $y$}$ and $\textrm{\boldmath $\Phi$}$.

Unfortunately, directly averaging the logarithm of random variables is, in general, technically difficult. 
Thus, we resort to the replica method to practically resolve this difficulty \cite{replica}. 
For this, we first evaluate the $n$-th moment of the marginal likelihood $\left [P^n \left(\vm{y} |\vm{\Phi} \right) \right]_{\vm{\Phi},\vm{y}}$ for $n=1,2,\ldots \in \mathbb{N}$ using the formula 
\begin{eqnarray}
&&P^n \left( \textrm{\boldmath $y$} |\vm{\Phi}\right)
=\int \prod_{a=1}^n \left (\textrm{d} \textrm{\boldmath $x$}^a P\left( \textrm{\boldmath $x^a$}\right)\right )
\prod_{a=1}^n \prod_{\mu=1}^M \Theta\left ( (\vm{y})_\mu  (\vm{\Phi} \vm{x}^a)_\mu \right ), 
\label{eq:expansion}
\end{eqnarray}
which holds only for $n=1,2,\ldots \in \mathbb{N}$. 
Here, $\vm{x}^a$ ($a=1,2,\ldots,n$) denotes  the $a$-th replicated signal. 
Averaging (\ref{eq:expansion}) with respect to $\vm{\Phi}$ and $\vm{y}$ results in {the} saddle-point evaluation concerning the macroscopic variables 
$q_{0a} =q_{a0}\equiv N^{-1} \vm{x}^0 \cdot \vm{x}^a$ and $q_{ab}=q_{ba} \equiv N^{-1} \vm{x}^a \cdot \vm{x}^b$ ($a,b =1,2,\ldots,n$). 

Although (\ref{eq:expansion}) holds only for $n \in \mathbb{N}$, the expression $N^{-1} \log \left [P^n \left(\vm{y} |\vm{\Phi}\right) \right]_{\vm{\Phi},\vm{y}}$ 
obtained by the saddle-point evaluation under a certain assumption concerning the permutation symmetry with respect to the replica indices $a,b=1,2,\ldots n$ is 
obtained as an analytic function of $n$, which is likely to also hold for $n \in \mathbb{R}$. Therefore, we next utilize the analytic function to evaluate 
the average of the logarithm of the partition function as 
\begin{eqnarray}
\bar{f}=-\lim_{n \to 0} (\partial /\partial n )N^{-1} \log \left [P^n \left(\vm{y} |\vm{\Phi}\right) \right]_{\vm{y},\vm{\Phi}}. 
\label{free energy density}
\end{eqnarray}
In particular, under the replica symmetric ansatz, where the dominant saddle-point is assumed to be of the form  
\begin{eqnarray}
q_{ab}=q_{ba}=\left \{
\begin{array}{ll}
\rho & (a=b=0) \cr
m & (a=1,2,\ldots,n; \ b=0) \cr
Q & (a=b=1,2,\ldots,n) \cr
q & (a\ne b =1,2,\ldots,n) 
\end{array}
\right . ,  
\label{RSanzats}
\end{eqnarray}
The above procedure expresses the average free energy density as 
\begin{eqnarray}
\bar{f} 
  &=& -\mathop{\rm extr}_{\omega} \!\Biggr\{ \int \textrm{d}x^0 P\left( x^0\right) \int\textrm{D}z\phi \left(\sqrt{\hat{q}}z+\hat{m}x^{0};\hat{Q}\right)
+\frac{1}{2}Q\hat{Q}+\frac{1}{2}q\hat{q}-m\hat{m} \nonumber\\
&&+2\alpha\int \textrm{D}t H\left(\frac{m}{\sqrt{\rho q-m^2}}t \right )\log H\left(\sqrt{\frac{q}{Q-q}}t\right) 
\Biggr\} .
\label{eq:free energy}
\end{eqnarray}
Here, $\alpha=M/N$,  {$H(x)=\int_x^{+\infty} {\rm D}z$, $\textrm{D}z=\textrm{d}z 
\exp \left (-z^2/2 \right )/\sqrt{2\pi}$ is a Gaussian measure,} $\textrm{extr}_{X}\{g(X)\}$ denotes the extremization of a function $g(X)$ with respect to $X$, 
$\omega = \{Q, q, m, \hat{Q}, \hat{q}, \hat{m}\}$, and
\begin{eqnarray}
\phi \left(\sqrt{\hat{q}}z+\hat{m}x^{0};\hat{Q}\right)\nonumber\\
=\log{\Bigg\{\int \textrm{d}x P\left( x\right)
  \exp{\left(-\frac{\hat{Q}+\hat{q}}{2}x^2+(\sqrt{\hat{q}}z+\hat{m}x^{0})x\right)}\Bigg\}}.
\end{eqnarray}
The derivation of $(\ref{eq:free energy})$ is provided in \ref{replicaderivation}.

In evaluating the right-hand side of (\ref{free energy density}), 
$P(\vm{y}|\vm{\Phi})$ not only gives the marginal likelihood (the partition function), but also 
the conditional density of $\vm{y}$ for taking the configurational average.
This accordance between the partition function and the distribution of the quenched random variables 
is generally known as the Nishimori condition in spin glass 
theory \cite{Nishimori}, 
 {for which the replica symmetric ansatz (\ref{RSanzats})
is supported by other schemes than the replica method \cite{RS_NL1, RS_NL2}, yielding the identity 
$\left [P^n \left(\vm{y} |\vm{\Phi}\right) \right]_{\vm{y},\vm{\Phi}}=\int {\rm d}\vm{\Phi} P(\vm{\Phi})
\left (\sum_{\vm{y}} P^{n+1}(\vm{y}|\vm{\Phi}) \right )$.
This} indicates that the true signal, $\vm{x}^0$, can be handled on an equal footing with 
the other $n$ replicated signals $\vm{x}^1,\vm{x}^2,\ldots,\vm{x}^n$ in the replica computation. 
As $n\to 0$, this higher replica symmetry among the $n+1$ replicated variables allows us to further simplify 
the replica symmetric ansatz (\ref{RSanzats}) by imposing four extra constraints:
$Q=\rho$, $q=m$, $\hat{Q}=0$, and $\hat{q}=\hat{m}$. 
As a consequence, the extremization condition of (\ref{eq:free energy}) is summarized by the non-linear equations
\begin{eqnarray}
m&=& 
\int {\rm D} t \frac{\left (\int {\rm d} x  xe^{-\frac{\hat{m}}{2}x^2+\sqrt{\hat{m}}t x} P(x) \right )^2}
{\int {\rm d} x  e^{-\frac{\hat{m}}{2}x^2+\sqrt{\hat{m}}t x} P(x)}\label{m_result}\\
\hat{m}&=&\frac{\alpha}{\pi\sqrt{2\pi}\left(\rho-m\right)}\int \textrm{d}t \frac{\exp{\left\{  -\frac{\rho+m}{2\left(\rho-m\right)}t^2\right\}}}{H\left(\sqrt{\frac{m}{\rho-m}}t\right)}.\label{mhat_result} 
\end{eqnarray} 

In physical terms, the value of $m$ determined by these equations is the typical overlap $N^{-1} \left [\vm{x}^0 \cdot \left \langle \vm{x} \right \rangle_{|\vm{y},\vm{\Phi}}\right]_{\vm{y},\vm{\Phi}} $ between the original signal $\vm{x}^0$ and 
the posterior mean $\left \langle \vm{x} \right \rangle_{|\vm{y},\vm{\Phi}}$. 
The law of large numbers and the self-averaging property guarantee that both $N^{-1} | \vm{x} |^2$ and $N^{-1}|\vm{x}^0|^2$
converge to $\rho$ with a probability of unity for typical samples.   
This indicates that the typical value of the direction cosine between $\vm{x}^0$ and $\hat{\vm{x}}^{\rm Bayes}(\vm{y})$
can be evaluated as 
$\left [ (\vm{x}^0 \cdot \hat{\vm{x}}^{\rm Bayes}(\vm{y}))/(|\vm{x}^0||\hat{\vm{x}}^{\rm Bayes}(\vm{y})| )\right ]_{\vm{y},\vm{\Phi}}
\simeq \left [ (\vm{x}^0 \cdot \left \langle \vm{x} \right \rangle_{|\vm{y},\vm{\Phi}} )\right ]_{\vm{y},\vm{\Phi}}/
\left (\left [ \left |\vm{x}^0 \right |\right ]_{\vm{x}^0} \left  [ \left |\left \langle \vm{x} \right \rangle_{|\vm{y},\vm{\Phi}} \right |\right ]_{\vm{y},\vm{\Phi}} \right )
=N m /(N\sqrt{\rho m})=\sqrt{m/\rho}$. 
Therefore, the MSE in (\ref{MSEdef}) can be expressed using $m$ and $\rho$ as
\begin{eqnarray}
{\rm MSE}
=2\left (1-\sqrt{\frac{m}{\rho}} \right ).
\label{MSE}
\end{eqnarray}

The symmetry between $\vm{x}^0$ and the other replicated variables $\vm{x}^a$ $(a=1,2\ldots,n)$ 
provides $\bar{f}$ with further information-theoretic meanings. 
Inserting $P(\vm{y},\vm{\Phi})=P(\vm{y}|\vm{\Phi})P(\vm{\Phi})$ into the definition of $\bar{f}$ gives 
$\bar{f}=N^{-1}\int {\rm d}\vm{\Phi} P(\vm{\Phi}) \left (-\sum_{\vm{y}}P(\vm{y}|\vm{\Phi})\log P(\vm{y}|\vm{\Phi}) \right )$, 
which indicates that $\bar{f}$ accords with the  entropy density of $\vm{y}$ for typical measurement matrices $\vm{\Phi}$. 
The expression $P(\vm{y} |\vm{x},\vm{\Phi})=\prod_{\mu=1}^M \Theta\left (y_\mu (\vm{\Phi}\vm{x})_\mu \right ) \in \{0,1\}$ guarantees
that the conditional entropy of $\vm{y}$ given $\vm{x}$ and $\vm{\Phi}$, 
$-\sum_{\vm{y}} P(\vm{y}|\vm{x},\vm{\Phi}) \log P(\vm{y}|\vm{x},\vm{\Phi}) $, always vanishes. 
These indicate that $\bar{f}$ also implies a mutual information density between $\vm{y}$ and $\vm{x}$. This physically quantifies the optimal information gain (per entry) of $\vm{x}$ that can be extracted from 
the 1-bit measurement $\vm{y}$ for typical $\vm{\Phi}$.

\section{
Bayesian optimal signal reconstruction by GAMP}
\label{BP}
Equation (\ref{MSE}) represents the potential performance of the Bayesian optimal signal reconstruction of 
1-bit CS. 
However, in practice, exploiting this performance is a non-trivial task, because performing 
the exact Bayesian reconstruction (\ref{Bayesian_reconstruction}) is computationally difficult. 
To resolve this difficulty, we now develop an approximate reconstruction algorithm 
following the framework of belief propagation (BP). 
Actually, BP has been 
successfully employed for standard CS problems 
{with} linear measurements, showing excellent performance in terms of both reconstruction accuracy and computational efficiency \cite{AMP}. 
To incorporate the non-linearity of the 1-bit measurement, we employ a variant of BP
known as generalized approximate message passing (GAMP) \cite{GAMP}, 
which can also be regarded as 
an approximate Bayesian inference algorithm for perceptron-type networks \cite{kabaUdaLPTN}.  

In general, the canonical BP equations for the probability measure $P(\vm{x}|\vm{\Phi},\vm{y})$ are expressed in terms of $2MN$ messages, 
$m_{i\rightarrow \mu}\left(x_i\right)$ and $m_{\mu\rightarrow i}\left(x_i\right) (i=1,2,\cdots,N; \mu=1,2,\cdots,M)$, 
which represent probability distribution functions that carry posterior information and output measurement information, respectively. 
They can be written as
\begin{eqnarray}
m_{\mu\rightarrow i}\left(x_i\right)
=\frac{1}{Z_{\mu\rightarrow i}}\int\prod_{j \neq i}\textrm{d} x_j P\left(y_{\mu}|u_{\mu}\right) \prod_{j \neq i} m_{j\rightarrow \mu}\left(x_j\right) \label{m_mu_i}\\
m_{i\rightarrow \mu}\left(x_i\right)
=\frac{1}{Z_{i\rightarrow \mu}} P\left(x_i\right)\prod_{\gamma \neq \mu}m_{\gamma\rightarrow i}\left(x_i\right)\label{m_i_mu}
\end{eqnarray}
Here, $Z_{\mu\rightarrow i}$ and $Z_{i\rightarrow \mu}$ are normalization factors ensuring that 
$\int\textrm{d} x_i m_{\mu\rightarrow i}(x_{i})=\int\textrm{d} x_i m_{i\rightarrow \mu}(x_{i})=1$, and 
we also define $u_{\mu}\equiv\left(\vm{\Phi}\vm{x}\right)_{\mu}$. 
Using (\ref{m_mu_i}), the {approximation of} marginal distributions $P(x_i|\vm{\Phi},\vm{y})=
\int \prod_{j\ne i} {\rm d}x_j P(\vm{x}|\vm{\Phi},\vm{y})$, which 
are often termed beliefs, are evaluated as
\begin{eqnarray}
m_{i}\left (x_i \right )
=\frac{1}{Z_i} P(x_i)\prod_{\mu=1}^M m_{\mu \rightarrow i} \left(x_i\right), \label{m_i}
\end{eqnarray}
where $Z_i$ is a normalization factor for $\int {\rm d} x_i m_i\left (x_i \right )=1$. 
To simplify the notation, we hereafter convert all  measurement results to $+1$ by multiplying  each row of the measurement matrix $\vm{\Phi}=(\Phi_{\mu i})$ by $y_{\mu}$ $(\mu=1,2,\ldots, N)$, giving $(\Phi_{\mu i}) \to (y_\mu \Phi_{\mu i})$, and  denote the resultant matrix as $\vm{\Phi}=(\Phi_{\mu i})$. In the new notation, $P\left(y_{\mu}|u_{\mu}\right)=\Theta\left( u_\mu\right)$.

Next, we introduce means and variances of $x_i$ in the posterior information message distributions as 
\begin{eqnarray}
a_{i\rightarrow\mu}\equiv \int\textrm{d}x_{i}x_{i}m_{i\rightarrow \mu}\left(x_i\right) \label{a_imu}\\
\nu_{i\rightarrow\mu}\equiv \int\textrm{d}x_{i}x_{i}^{2}m_{i\rightarrow \mu}\left(x_i\right)-a_{i\rightarrow\mu}^2 \label{nu_imu}.
\end{eqnarray}
We also define 
$\omega_{\mu}\equiv \sum_{i}\Phi_{\mu i}a_{i\rightarrow\mu}$ and 
$V_{\mu}\equiv \sum_{i}\Phi_{\mu i}^2 \nu_{i\rightarrow\mu}$ for notational convenience. 
Similarly, the means and variances  of the beliefs, $a_i$ and $\nu_i$, are introduced as
$a_{i}\equiv \int\textrm{d}x_{i}x_{i}m_{i}\left(x_i\right)$ and 
$\nu_{i}\equiv \int\textrm{d}x_{i}x_{i}^{2}m_{i}\left(x_i\right)-a_{i}^2 $.
Note that $\vm{a}=(a_i)$ represents the {approximation of the} posterior mean $\left \langle \vm{x} \right \rangle_{|\vm{y},\vm{\Phi}}$. 
This, in conjunction with a consequence of the law of large numbers $\left \langle \vm{x}/|\vm{x} |\right \rangle_{|\vm{y},\vm{\Phi}}
\simeq \left \langle \vm{x} \right \rangle_{|\vm{y},\vm{\Phi}}/\sqrt{N\rho}$, indicates that 
the Bayesian optimal reconstruction is approximately performed as $\hat{\vm{x}}^{\rm Bayes}(\vm{y})\simeq \sqrt{N\rho} \vm{a}/|\vm{a}|$.

To enhance the computational tractability, let us rewrite the functional equations of 
(\ref{m_mu_i}) and (\ref{m_i_mu}) {into} algebraic equations using sets of $a_{i\to\mu}$ and $\nu_{i \to \mu}$.
To do this, we insert the identity
\begin{eqnarray}
1&=&\int\textrm{d}u_{\mu} \delta\left(u_{\mu}-\sum_{i=1}^{N}\Phi_{\mu i}x_i\right)\nonumber\\
 &=&\int\textrm{d}u_{\mu} \frac{1}{2\pi}\int\textrm{d}\hat{u}_{\mu}\exp{\Biggl\{-i\hat{u}_{\mu}\left(u_{\mu}-\sum_{i=1}^{N}\Phi_{\mu i}x_i\right)\Biggr\}}
\end{eqnarray}
into (\ref{m_mu_i}), which yields
\begin{eqnarray}
m_{\mu\rightarrow i}\left(x_i\right)
=\frac{1}{2\pi Z_{\mu\rightarrow i}} \int\textrm{d}u_{\mu}P\left(y_{\mu}|u_{\mu}\right) 
 \int\textrm{d}\hat{u}_{\mu}\exp{\Biggl\{-i\hat{u_{\mu}}\left(u_{\mu}-\Phi_{\mu i}x_i\right)\Biggr\}}\nonumber\\
\times\prod_{j \neq i}\Biggl\{\int\textrm{d}x_{j}m_{j\rightarrow \mu}\left( x_j\right)\exp{\Bigl\{ i \hat{u}_{\mu}\Phi_{\mu j}x_j\Bigr\}} \Biggr\}.
\label{m_mu_i_1}
\end{eqnarray}
The smallness of $\Phi_{\mu i}$ allows us to truncate the Taylor series of the last exponential in {equation} (\ref{m_mu_i_1}) 
up to the second order of $i \hat{u}_{\mu}\Phi_{\mu j}x_j$. 
Integrating  $\int {\rm d}x_j m_{j\to \mu}(x_j) \left (\ldots \right )$ for $j \ne i$, 
we obtain the expression 
\begin{eqnarray}
m_{\mu\rightarrow i}\left(x_i\right)
=\frac{1}{2\pi Z_{\mu\rightarrow i}} \int\textrm{d}u_{\mu}P\left(y_{\mu}|u_{\mu}\right) 
\int\textrm{d}\hat{u}_{\mu}\exp{\Biggl\{-i\hat{u_{\mu}}\left(u_{\mu}-\Phi_{\mu i}x_i\right)\Biggr\}}\nonumber\\
\times\exp{\Biggl\{ i\hat{u}_{\mu}(
\omega_\mu -\Phi_{\mu i}a_{i\to \mu}
)
-\frac{\hat{u}_{\mu}^2}{2}(
V_\mu -\Phi_{\mu i}^2\nu_{i\to \mu}
)\Biggr\}}, 
\label{m_mu_i_4}
\end{eqnarray}
and carrying out the resulting Gaussian {intergral} of $\hat{u}_{\mu}$, we obtain
\begin{eqnarray}
m_{\mu\rightarrow i}\left(x_i\right)
&=&\frac{1}{Z_{\mu\rightarrow i}
\sqrt{2\pi
(V_\mu -\Phi_{\mu i}^2 \nu_{i \to \mu})
}}\int\textrm{d}u_{\mu}P\left(y_{\mu}|u_{\mu}\right)\nonumber\\
&&\times\exp{\Biggl\{-\frac{\left( u_{\mu}-\omega_\mu -\Phi_{\mu i}(x_i -a_{i\rightarrow\mu})\right)^2}
{2(V_\mu -\Phi_{\mu i}^2 \nu_{i \to \mu})}\Biggr\}}. 
\label{m_mu_i_2}
\end{eqnarray}
Since $\Phi_{\mu i}^2$ vanishes as 
$O(N^{-1})$ while $\nu_{i \to \mu} \sim O(1)$, we can omit $\Phi_{\mu i}^2 \nu_{i \to \mu}$  in (\ref{m_mu_i_2}). 
In addition, we replace $\Phi_{\mu j}^2$ in $V_{\mu}= \sum_{i}\Phi_{\mu j}^2\nu_{i\rightarrow \mu}$ 
with its expectation $N^{-1}$, utilizing the law of large numbers.
This removes the dependence on the index $\mu$, making all $V_{\mu}$ equal to their average 
\begin{equation}
V\equiv\frac{1}{N}\sum_{i=1}^{N}\nu_i.
\label{V}
\end{equation}
The smallness of $\Phi_{\mu i}(x_i -a_{i\rightarrow\mu})$ again allows us to truncate the 
Taylor series of the exponential in (\ref{m_mu_i_2}) up to the second order. 
Thus, we have a parameterized expression of $m_{\mu\rightarrow i}\left(x_i\right)$:
\begin{eqnarray}
m_{\mu\rightarrow i}\left(x_i\right)
\propto \exp{\Biggl\{ -\frac{A_{\mu\rightarrow i}}{2}x_{i}^2 +B_{\mu\rightarrow i}x_i \Biggr\}}, 
\label{m_mu_i_3}
\end{eqnarray}
where the parameters $A_{\mu \rightarrow i}$ and $B_{\mu\rightarrow i}$ are evaluated as
\begin{eqnarray}
A_{\mu\rightarrow i}=(g_{\rm out}^{\prime})_{\mu}\Phi_{\mu i}^2\label{A}\\
B_{\mu\rightarrow i}=(g_{\rm out})_{\mu}\Phi_{\mu i}+(g_{\rm out}^{\prime})_{\mu}\Phi_{\mu i}^2 a_{i\rightarrow \mu}\label{B}
\end{eqnarray}
using 
 \begin{eqnarray}
(g_{\rm out})_{\mu} &\equiv& \frac{\partial }{\partial \omega_\mu} 
\log \left (\int {\rm d} u_\mu P(y_\mu|u_\mu)\exp \left (-\frac{(u_\mu-\omega_\mu)^2}{2V} \right ) \right )
\label{g_out} \\
(g^\prime_{\rm out})_{\mu} &\equiv & -\frac{\partial^2 }{\partial \omega_\mu^2} 
\log \left (\int {\rm d} u_\mu P(y_\mu|u_\mu)\exp \left (-\frac{(u_\mu-\omega_\mu)^2}{2V} \right ) \right ).
\label{g_out_p}
\end{eqnarray}
The derivation of these is given in \ref{GAMP_derivation}. 
Equations (\ref{A}) and (\ref{B}) act as the algebraic expression of (\ref{m_mu_i}). 
In the sign output channel, inserting $P\left(y_{\mu}|u_{\mu}\right)=\Theta\left( u_\mu\right)$ into (\ref{g_out}) 
gives $(g_{\rm out})_{\mu}$ and $(g_{\rm out}^{\prime})_{\mu}$ for 1-bit CS as
\begin{eqnarray}
(g_{\rm out})_{\mu}=\frac{\exp{\left( -\frac{\omega_{\mu}^2}{2V} \right)}}
{\sqrt{2\pi V}H\left( -\frac{\omega_{\mu}}{\sqrt{V}}\right)}  \label{g_out_1bit}\\
(g_{\rm out}^{\prime})_{\mu}=(g_{\rm out})_{\mu}^2 + \frac{\omega_{\mu}}{V}(g_{\rm out})_{\mu}\label{g_out_p_1bit}.
\end{eqnarray}

To obtain a similar expression for (\ref{m_i_mu}), we substitute the last expression of (\ref{m_mu_i_3})
into (\ref{m_i_mu}), which leads to 
\begin{eqnarray}
m_{i\rightarrow\mu}(x_i)=\frac{1}{\tilde{Z}_{i\rightarrow\mu}}\left[(1-\rho)\delta(x_i)+\rho \tilde{P}(x_i) \right]
 e^{-(x_{i}^2 /2)\sum\limits_{\gamma\neq\mu}A_{\gamma\rightarrow i} +x_i\sum\limits_{\gamma\neq\mu}B_{\gamma\rightarrow i} }.
\end{eqnarray}
This indicates that $\prod_{\gamma \neq \mu}m_{\gamma\rightarrow i}\left(x_i\right)$ in (\ref{m_i_mu}) can 
be expressed as a Gaussian distribution with 
mean $(\sum_{\gamma\neq\mu}B_{\gamma\rightarrow i})/(\sum_{\gamma\neq\mu}A_{\gamma\rightarrow i})$ and 
variance $(\sum_{\gamma\neq\mu}A_{\gamma\rightarrow i})^{-1}$.
Inserting these into (\ref{a_imu}) and (\ref{nu_imu}) provides the algebraic expression of
(\ref{m_i_mu}) as 
\begin{eqnarray}
a_{i\rightarrow\mu}=f_{a}\left( \frac{1}{\sum_{\gamma\neq\mu}A_{\gamma\rightarrow i}}, \frac{\sum_{\gamma\neq\mu}B_{\gamma\rightarrow i}}{\sum_{\gamma\neq\mu}A_{\gamma\rightarrow i}}\right)  \label{VstepA},\\
\nu_{i\rightarrow\mu}=f_{c}\left( \frac{1}{\sum_{\gamma\neq\mu}A_{\gamma\rightarrow i}}, \frac{\sum_{\gamma\neq\mu}B_{\gamma\rightarrow i}}{\sum_{\gamma\neq\mu}A_{\gamma\rightarrow i}}\right) \label{VstepNu},  
\end{eqnarray}
 {
where $f_a(\Sigma^2,R)$ and $f_c(\Sigma^2,R)$ stand for the mean and variance of 
an auxiliary distribution of $x$
\begin{eqnarray}
{\cal M}
(x|\Sigma^2,R)
=\frac{1}{{\cal Z}(\Sigma^2,R)} 
\left[(1-\rho)\delta(x)+\rho \tilde{P}(x) \right]
 \frac{1}{\sqrt{2\pi\Sigma^2}}e^{-\frac{(x-R)^2}{2\Sigma^2}}
\end{eqnarray}
where ${\cal Z}(\Sigma^2,R)$ is a normalization constant, 
respectively. 
For instance, when $\tilde{P}(x)$ is a Gaussian distribution of mean $\bar{x}$ and variance $\sigma^2$, 
we have }
\begin{eqnarray}
f_{a}(\Sigma^2 , R) = \frac{\bar{x}\Sigma^2+R\sigma^2}
                    {\frac{(1-\rho)(\sigma^2+\Sigma^2)^{3/2}}{\rho \Sigma} \exp{\Bigl\{-\frac{R^2}{2\Sigma^2}+\frac{(R-\bar{x})^2}{2(\sigma^2+\Sigma^2)}\Bigr\}}+(\sigma^2+\Sigma^2)}\label{fa},\\
f_{c}(\Sigma^2 , R)=\Biggl\{ \rho(1-\rho)\Sigma\left(\sigma^2+\Sigma^2\right)^{-5/2}\left[\sigma^2\Sigma^2\left(\sigma^2+\Sigma^2\right)+\left(\bar{x}\Sigma^2+R\sigma^2\right)^2\right]\nonumber\\
                    \times\exp{\Bigl\{-\frac{R^2}{2\Sigma^2}-\frac{(R-\bar{x})^2}{2(\sigma^2+\Sigma^2)}\Bigr\}}
                    +\rho^2\exp{\Bigl\{-\frac{(R-\bar{x})^2}{\sigma^2+\Sigma^2}\Bigr\}}\frac{\sigma^2\Sigma^4}{(\sigma^2+\Sigma^2)^2}\Biggr\} \nonumber\\
                    \times\Biggl\{(1-\rho)\exp{\Bigl\{-\frac{R^2}{2\Sigma^2}\Bigr\}}+\rho\frac{\Sigma}{\sqrt{\sigma^2+\Sigma^2}}\exp{\Bigl\{ -\frac{(R-\bar{x})^2}{2(\sigma^2+\Sigma^2)}\Bigr\}} \Biggr\}^{-2}.
                    \label{fc}
\end{eqnarray}

For the signal reconstruction, we need to evaluate the moments of $m_i(x_i)$. 
This can be performed by simply adding back the $\mu$ dependent part to (\ref{VstepA}) and (\ref{VstepNu}) as
\begin{eqnarray}
a_i=f_a (\Sigma^2_i , R_i) , \label{a}\\
\nu_i=f_c (\Sigma^2_i , R_i), \label{nu}
\end{eqnarray}
where $\Sigma^2_i= \left(\sum_{\mu}A_{\mu\rightarrow i}\right)^{-1}$, $R_i= 
\frac{\sum_{\mu}B_{\mu\rightarrow i}}{\sum_{\mu}A_{\mu\rightarrow i}}$.
For large $N$, 
$\Sigma^{2}_{i}$ typically converges to a constant, independent of the index, as $\Sigma^2$. 
This, in conjunction with (\ref{A}) and (\ref{B}), yields 
\begin{eqnarray}
\Sigma^2=\left( \frac{1}{N}\sum_{\mu}(g_{\rm out}^{\prime})_{\mu} \right)^{-1}, \label{Sigma}\\
R_i=\left( \sum_{\mu}(g_{\rm out})_{\mu}\Phi_{\mu i}\right)\Sigma^2 +a_i.\label{R}
\end{eqnarray}

BP updates $2MN$ messages using (\ref{A}), (\ref{B}), (\ref{VstepA}), and 
(\ref{VstepNu}) ($i=1,2,\cdots N, \mu=1,2,\cdots M$) in each iteration.
This requires a computational cost of $O(M^2 \times N + M \times N^2)$ per iteration, 
which may limit the practical utility of BP to systems of relatively small size.  
To enhance the practical utility, let us rewrite the BP equations {into} those of $M+N$ messages for large $N$, which 
will result in a significant reduction of computational complexity to $O(M\times N)$ per iteration.
To do this, 
we express $a_{i\to \mu}$ by applying Taylor's expansion to (\ref{VstepA})  around $R_i$ as
\begin{eqnarray}
a_{i\rightarrow\mu}
&=& f_{a}\left( \frac{1}{\sum_{\gamma}A_{\gamma\rightarrow i}-A_{\mu\rightarrow i}}, \frac{\sum_{\gamma}B_{\gamma\rightarrow i}-B_{\mu\rightarrow i}}{\sum_{\gamma}A_{\gamma\rightarrow i}-A_{\mu\rightarrow i}}\right)\nonumber\\
&\simeq & a_i 
+\frac{\partial f_{a}(\Sigma^2,R_i)}{\partial R_i} (-B_{\mu\rightarrow i}\Sigma^2)+ O(N^{-1}), 
\label{TAP}
\end{eqnarray}
where $B_{\mu\rightarrow i}\sim O(N^{-1/2})$  and 
$\sum_{\gamma}A_{\gamma\rightarrow i}-A_{\mu\rightarrow i}$ is approximated as $\sum_{\gamma}A_{\gamma\rightarrow i}=\Sigma^{-2}$,
because of the smallness of $A_{\mu\rightarrow i}\propto \Phi_{\mu i}^2 \sim O(N^{-1})$.   
Multiplying this by $\Phi_{\mu i}$  and summing  the resultant expressions over $i$
yields
\begin{eqnarray}
\omega_{\mu}
= \sum_{i} \Phi_{\mu i}a_i -(g_{\rm out})_{\mu}V, \label{omega}
\end{eqnarray}
where we have used $\nu_i =f_c=\Sigma^2 \frac{\partial f_a}{\partial R_i}$, 
which can be confirmed by (\ref{fa}) and (\ref{fc}).

Let us assume that $\{(a_i,\nu_i)\}$ and 
$\left \{\left ((g_{\rm out})_\mu, (g_{\rm out}^\prime)_\mu \right )\right \}$ 
are initially set to certain values. Inserting these into (\ref{V}) and (\ref{omega}) gives
$V$ and $\{\omega_\mu\}$. Substituting these into {equations} (\ref{g_out_1bit}) and (\ref{g_out_p_1bit})
yields a set of $\left \{\left ((g_{\rm out})_\mu, (g_{\rm out}^\prime)_\mu \right )\right \}$, which, in conjunction with $\{a_i\}$, offers
$\Sigma^2$ and $\{R_i\}$ through (\ref{Sigma}) and (\ref{R}). 
Inserting these into (\ref{a}) and (\ref{nu}) offers a new set of $\{(a_i,\nu_i)\}$. 
In this way, the iteration of (\ref{V}), (\ref{omega}) $\to $ (\ref{g_out_1bit}), (\ref{g_out_p_1bit})
$\to $ (\ref{Sigma}), (\ref{R}) $\to $ (\ref{a}), (\ref{nu}) $\to $  (\ref{V}), (\ref{omega}) $\to \ldots$
constitutes a closed set of equations to update the sets of $\{(a_i,\nu_i)\}$ 
and $\left \{\left ((g_{\rm out})_\mu, (g_{\rm out}^\prime)_\mu \right )\right \}$. 
This is the generic GAMP algorithm given a likelihood function $P(y|u)$ and   
a prior distribution $P(x)$ \cite{GAMP}. 

We term the entire procedure 
{the Approximate Message Passing for 1-bit
Compressed Sensing (1bitAMP) algorithm.} The pseudocode of this algorithm 
is summarized in Figure \ref{proposedalgorithm}. 
Three issues are noteworthy. 
First, for relatively large systems, e.g., $N =1024$, the iterative procedure converges easily  in most cases. 
Nevertheless, since it relies on the law of large numbers, 
some divergent behavior appears as $N$ becomes smaller. 
Even for such cases, however, employing an appropriate damping factor in conjunction with 
a normalization of $|\vm{a}|$ at each update considerably improves the convergence property. 
Second, the most time-consuming parts of this iteration are the matrix-vector multiplications 
$\sum_{\mu}(g_{\rm out})_{\mu}\Phi_{\mu i}$ in (\ref{R}) and 
$\sum_{i} \Phi_{\mu i} a_i$ in (\ref{omega}). This indicates that
the computational complexity is $O(NM)$ per iteration. 
Finally, $a_i$ in {equation} (\ref{R}) and $(g_{\rm out})_{\mu}V$ in {equation} (\ref{omega}) correspond to 
what is known as the {\em Onsager reaction term} in the spin glass literature  \cite{Onsager1,Onsager2}.
These terms stabilize the convergence of {1bitAMP}, 
effectively canceling the self-feedback effects.

\begin{figure}
\renewcommand{\thepseudocode}{\arabic{pseudocode}}
\setcounter{pseudocode}{0}
\begin{pseudocode}[ruled]
{Approximate Message Passing for 1-bit CS}
{\mathbf{a}^*, \mathbf{\nu}^*, \mathbf{\omega}^*}

1)\ \mbox{\bf Initialization}:\\
 \hspace{15pt}\text{a seed}: \hspace{65pt} \mathbf{a}_{0} \GETS \mathbf{a}^*\\
 \hspace{15pt}\text{$\nu$ seed}: \hspace{65pt} \mathbf{\nu}_0 \GETS \mathbf{\nu}^* \\
 \hspace{15pt}\text{$\omega$ seed}: \hspace{65pt} \mathbf{\omega}_0 \GETS \mathbf{\omega}^* \\
 \hspace{15pt}\text{Counter}: \hspace{62pt} k\GETS 0\\

2)\ \mbox{\bf Counter increase}:\\
 \hspace{30pt}k \GETS k+1\\

3)\ \mbox{\bf Mean of variances of posterior information message distributions}:\\
 \hspace{30pt}\mathbf{V}_{k}\GETS  \textsc{N}^{-1}(\text{sum} (\mathbf{\nu}_{k-1}))\vm{1}\\
4)\ \mbox{\bf Self-feedback cancellation}:\\
 \hspace{30pt}
 \mathbf{\omega}_{k} \GETS \mathbf{\Phi}\mathbf{a}_{k-1}-\mathbf{V}_{k}g_{\rm out}(\mathbf{\omega}_{k-1}, \mathbf{V}_{k})\\

5)\ \mbox{\bf Variances of output information message distributions}:\\
 \hspace{30pt}\Sigma_k^2 \GETS \textsc{N}(\text{sum} (g_{\rm out}^{\prime}(\mathbf{\omega}_{k}, \mathbf{V}_{k})))^{-1}\\

6)\ \mbox{\bf Average of output information message distributions}:\\ 
 \hspace{30pt}(\mathbf{R})_k\GETS \mathbf{a}_{k-1}+(g_{\rm out}(\mathbf{\omega}_{k}, \mathbf{V}_{k})\mathbf{\Phi})\Sigma_k^2\\

7)\ \mbox{\bf Posterior mean}:\\
 \hspace{30pt}{\mathbf{a}}_{k}\GETS f_a (\Sigma_{k}^2\vm{1}, \mathbf{R}_{k})\\

8)\ \mbox{\bf Posterior variance}:\\
 \hspace{30pt}{\mathbf{\nu}}_{k}\GETS f_c (\Sigma_{k}^2\vm{1}, \mathbf{R}_{k})\\

 
9)\ \bold{Iteration}: \mbox{Repeat from step 2 until convergence.}
\end{pseudocode}
\protect
\caption{\protect\label{proposedalgorithm}Pseudocode for 
{1-bitAMP}.
$\mathbf{a}^*$, $\mathbf{\nu}^*$, and $\mathbf{\omega}^*$ are the convergent vectors of $\mathbf{a}_k$, $\mathbf{\nu}_k$, and $\mathbf{\omega}_k$ obtained in the previous loop.
 $\vm{1}$ is the $N$-dimensional vector whose entries are all unity. }
\label{algorithm}
\end{figure}

\section{Results}

\begin{figure}[t]
  \begin{center}
    \begin{tabular}{c}
    
      \begin{minipage}{0.5\hsize}
        \begin{center}
          \includegraphics[clip, angle=0,width=8cm]{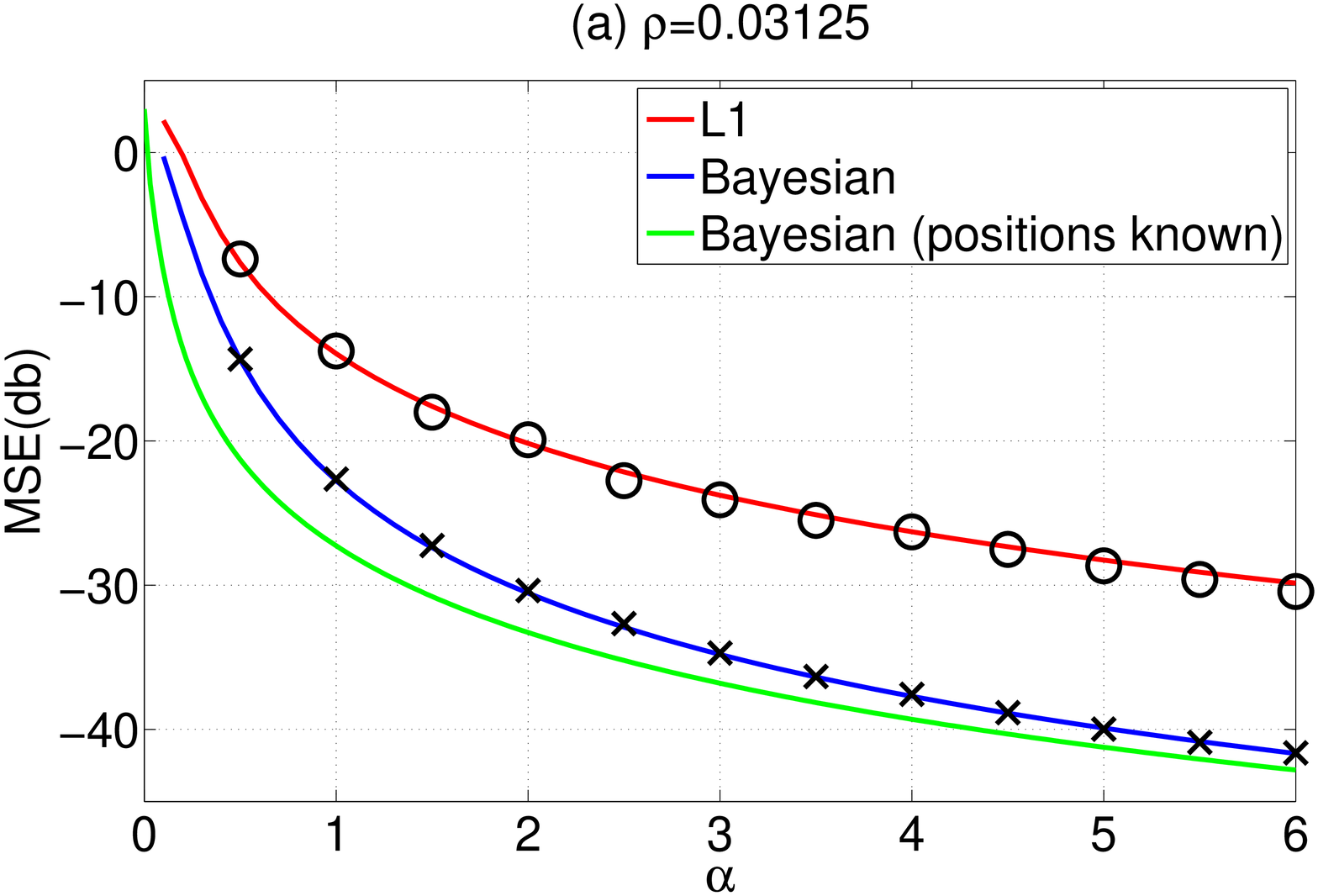}
          \hspace{1.6cm} 
        \end{center}
      \end{minipage}
      
      \begin{minipage}{0.5\hsize}
        \begin{center}
          \includegraphics[clip, angle=0,width=8cm]{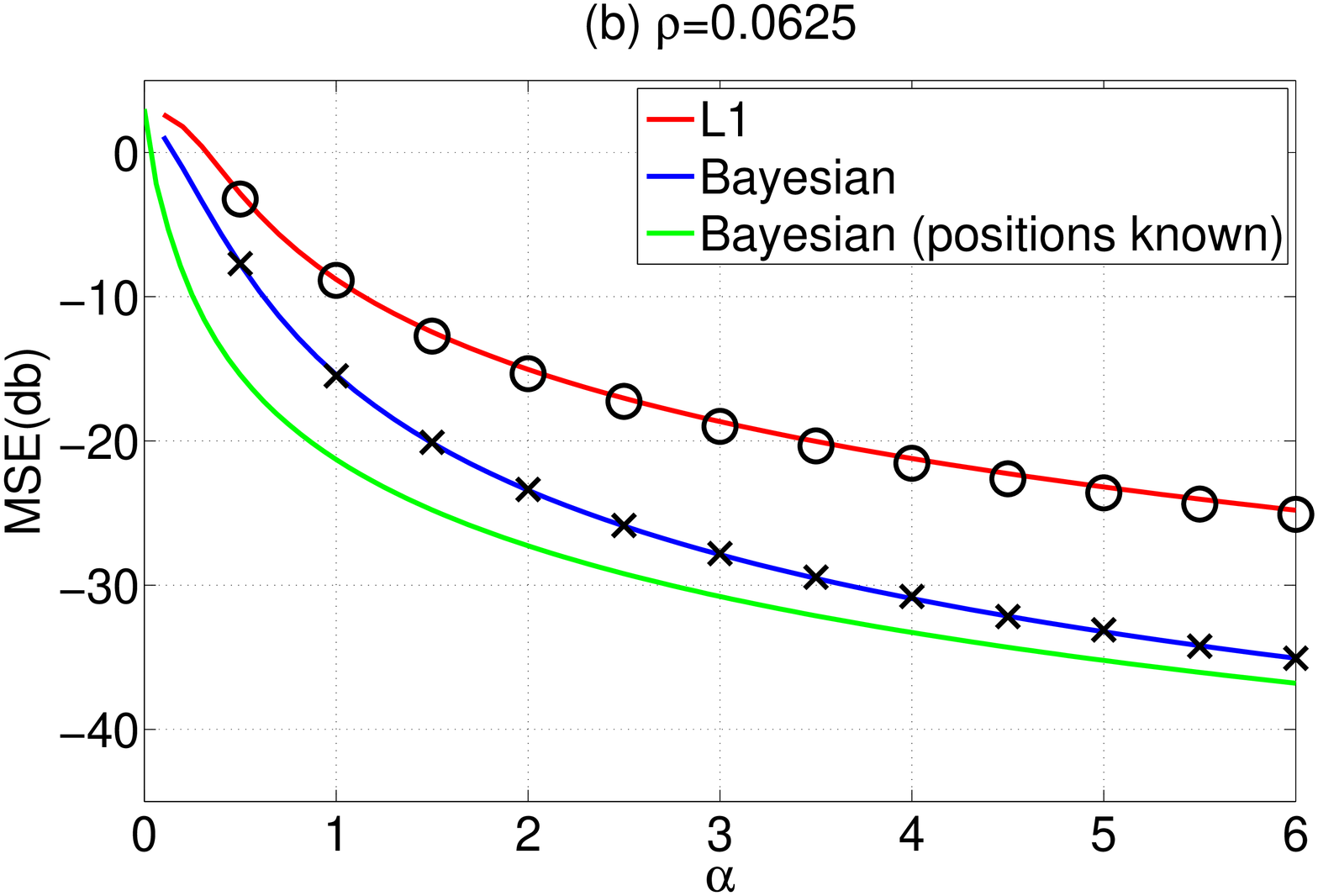}
          \hspace{1.6cm}
        \end{center}
      \end{minipage} 
\\
      \begin{minipage}{0.5\hsize}
        \begin{center}
          \includegraphics[clip, angle=0,width=8cm]{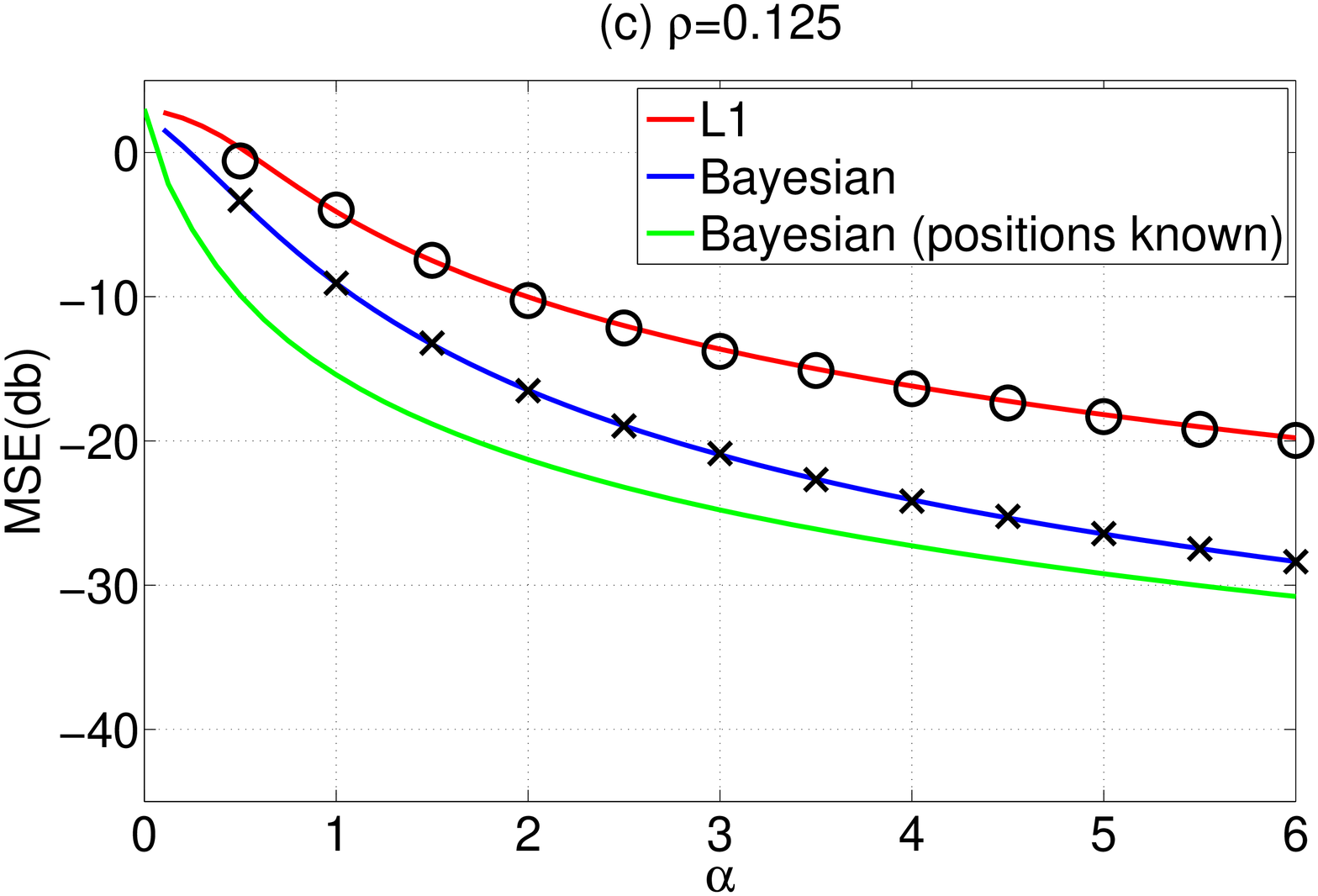}
          \hspace{1.6cm} 
        \end{center}
      \end{minipage}
      
      \begin{minipage}{0.5\hsize}
        \begin{center}
          \includegraphics[clip, angle=0,width=8cm]{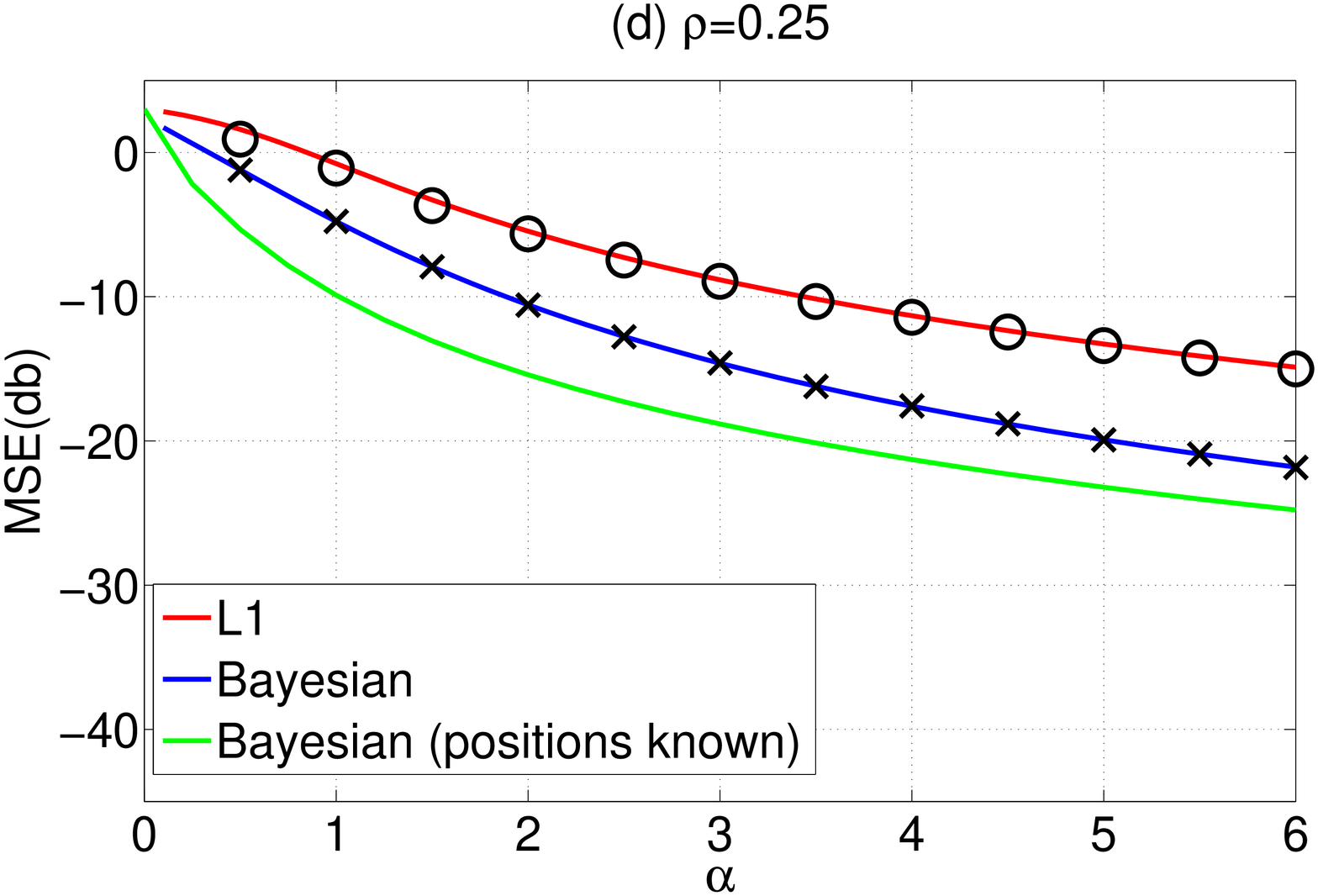}
          \hspace{1.6cm} 
        \end{center}
      \end{minipage}
    \end{tabular}
\caption{\label{fig:BPMSE}MSE (in decibels) versus measurement bit ratio $\alpha$ for 1-bit CS
 {for Gauss-Bernoulli prior}.
(a), (b), (c), and (d) correspond to $\rho=0.03125, 0.0625, 0.125$, and $0.25$, respectively.
Red curves represent the theoretical prediction of $l_1$-norm minimization \cite{YingKaba1bitCS2013};
blue curves represent the theoretical prediction of the Bayesian optimal approach;
green curves represent the theoretical prediction of the Bayesian optimal approach 
when the positions of all nonzero components in the signal are known, 
which is obtained by setting $\alpha\rightarrow\alpha/\rho$ and $\rho\rightarrow 1$ in (\ref{m_result}) and (\ref{mhat_result}).
Crosses represent the average of 1000 experimental results by the {1bitAMP} algorithm in Figure \ref{algorithm} for a system size of $N=1024$. {Circles show the average of 1000 experimental
results by an $l_1$-based algorithm RFPI proposed in  
\cite{1bitCS} for 1-bit CS in the system size of
$N = 128$. Although the replica symmetric prediction for the $l_1$-based approach is thermodynamically unstable,
the experimental results of RFPI are numerically consistent with it very well.}}
\label{figure2}
\end{center}
\end{figure}

\begin{figure}[ht]
  \begin{center}
          \includegraphics[clip, angle=0,width=10cm,keepaspectratio]{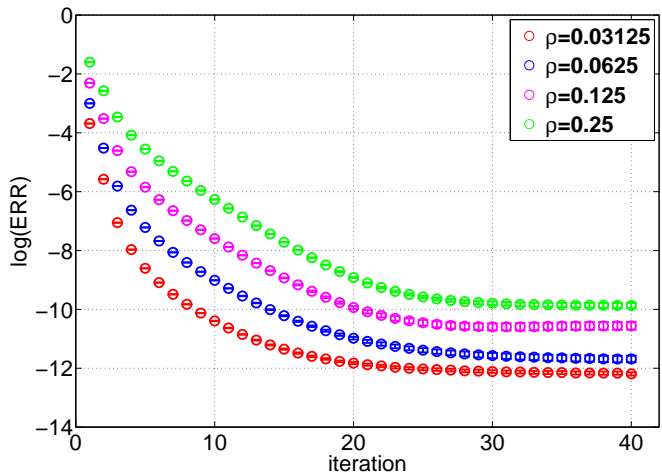}
        \end{center}
\caption{
Mean square differences (ERR) between estimated signals of two successive iterative update
of 1bitAMP for a signal size of $N=1024$ and $\alpha=6$, and the errorbar, which are evaluated from  {10000} experiments. 
Red, blue, magenta, and green represent $\rho=0.03125, 0.0625, 0.125$, and $0.25$, respectively.}
\label{ERR_iteration}
\end{figure}

To examine the utility of {1bitAMP}, we carried out numerical experiments  {for 
Gauss-Bernoulli prior,}
\begin{equation}
P\left(x\right) =\left(1-\rho\right) \delta \left( x \right)
+ \frac{\rho}{\sqrt{2\pi}}e^{-\frac{1}{2}x^2}
\end{equation}
with system size $N=1024$. 
We set initial conditions of $\vm{a}=0\vm{1}, \vm{\nu}=\rho\vm{1}$, and $\vm{\omega}=\vm{1}$, 
where $\vm{1}$ is the $N$-dimensional vector whose entries are all unity, 
and stopped the algorithm after $20$ iterations (Figure \ref{ERR_iteration}). 
The MSE results for various sets of $\alpha$ and $\rho$ 
are shown as crosses in Figures \ref{figure2} (a)--(d). 
Each cross denotes an experimental estimate obtained 
from 1000 experiments. The standard deviations are omitted, as they are smaller than the size of the symbols. 
The convergence time is short, which verifies the significant computational efficiency of {1bitAMP}.
For example, in a MATLAB$^{\textregistered}$ environment, for $\alpha=3, \rho=0.0625$, 
one experiment takes around 0.2 s. 

To test the consistency of {1bitAMP} with respect to replica theory, 
we solved the saddle-point equations (\ref{m_result}) and (\ref{mhat_result})  {for 
Gauss-Bernoulli prior} for each 
set of $\alpha$ and $\rho$.
The blue curves in Figures \ref{figure2} (a)--(d) show the theoretical MSE evaluated by (\ref{MSE}) against 
 $\alpha$ for $\rho=0.03125, 0.0625, 0.125$, and $0.25$. 
The excellent agreement between the numerical experiments and the theoretical prediction indicates 
that {1bitAMP} nearly saturates the potentially achievable MSE of the signal recovery scheme based on the Bayesian optimal approach.

For comparison, Figures \ref{figure2} (a)--(d) also plot 
the replica symmetric 
prediction of MSEs for the $l_1$-norm minimization approach (red curves)  {to the 
Gauss-Bernoulli signal}, 
which was examined in an earlier study~\cite{YingKaba1bitCS2013}.
Although the replica symmetric prediction is thermodynamically unstable, 
it is numerically consistent with the experimental results {(circles)} given by the algorithm proposed in \cite{1bitCS}. 
Therefore, the prediction at least serves as a good approximation.

We also plot the MSEs of the Bayesian optimal approach when the positions of the non-zero components of $\vm{x}$ are known 
(green curves). 
These act as lower bounds for the  MSEs of the Bayesian optimal approach. 
When the positions of non-zero components of $\vm{x}$ are known, we need not consider the part containing zero components. 
Therefore, the problem can be seen as that defined  
when a $\rho N$-dimensional signal $\vm{x}$ is measured by an $\alpha N\times \rho N$-dimensional matrix. 
In such situations, performance can be evaluated by setting $\rho =1$ and replacing 
$\alpha $ with $\alpha/\rho$ in (\ref{m_result}) and (\ref{mhat_result}), as the dimensionality of $\vm{x}$ is reduced from $N$ to $N\rho$. 
Solving (\ref{m_result}) and (\ref{mhat_result}) for $\alpha \gg 1$ shows that 
the MSEs of the Bayesian optimal approach can be asymptotically expressed as
\begin{eqnarray}
{\rm MSE^{Bayes}}  \simeq \frac{1.9258 \rho^2}{\alpha^2} = 1.9258 \times \left (\frac{N\rho}{M} \right )^2
\label{asymptotics}
\end{eqnarray}
for $\alpha \gg 1$, which accords exactly with the asymptotic form of the green curves (Figure~\ref{fig.asymptotic}: left panel, see \ref{app3}). 
 {Since we defined MSE with the normalized signal, 
this holds for all zero mean Gauss-Bernoulli distributions of any variance. }
On the other hand, 
the asymptotic form of the MSE for the $l_1$-norm approach is evaluated as
\begin{eqnarray}
{\rm MSE}^{l_1}
\simeq \frac{\pi^2\left[2(1-\rho)H\left (1/\sqrt{\hat{q}_{l_1}^\infty(\rho)} \right )+\rho \right]^2}{\alpha^2},
\label{MSE_L1}
\end{eqnarray}
where $\hat{q}_{l_1}^\infty(\rho)$ is the value of $\hat{q}$  for the $l_1$-norm approach obtained for $\alpha \to \infty$ (see \ref{app4}).

\begin{figure}[t]
  \begin{center}
    \begin{tabular}{c}
    
      \begin{minipage}{0.5\hsize}
        \begin{center}
          \includegraphics[clip, angle=0,width=8cm]{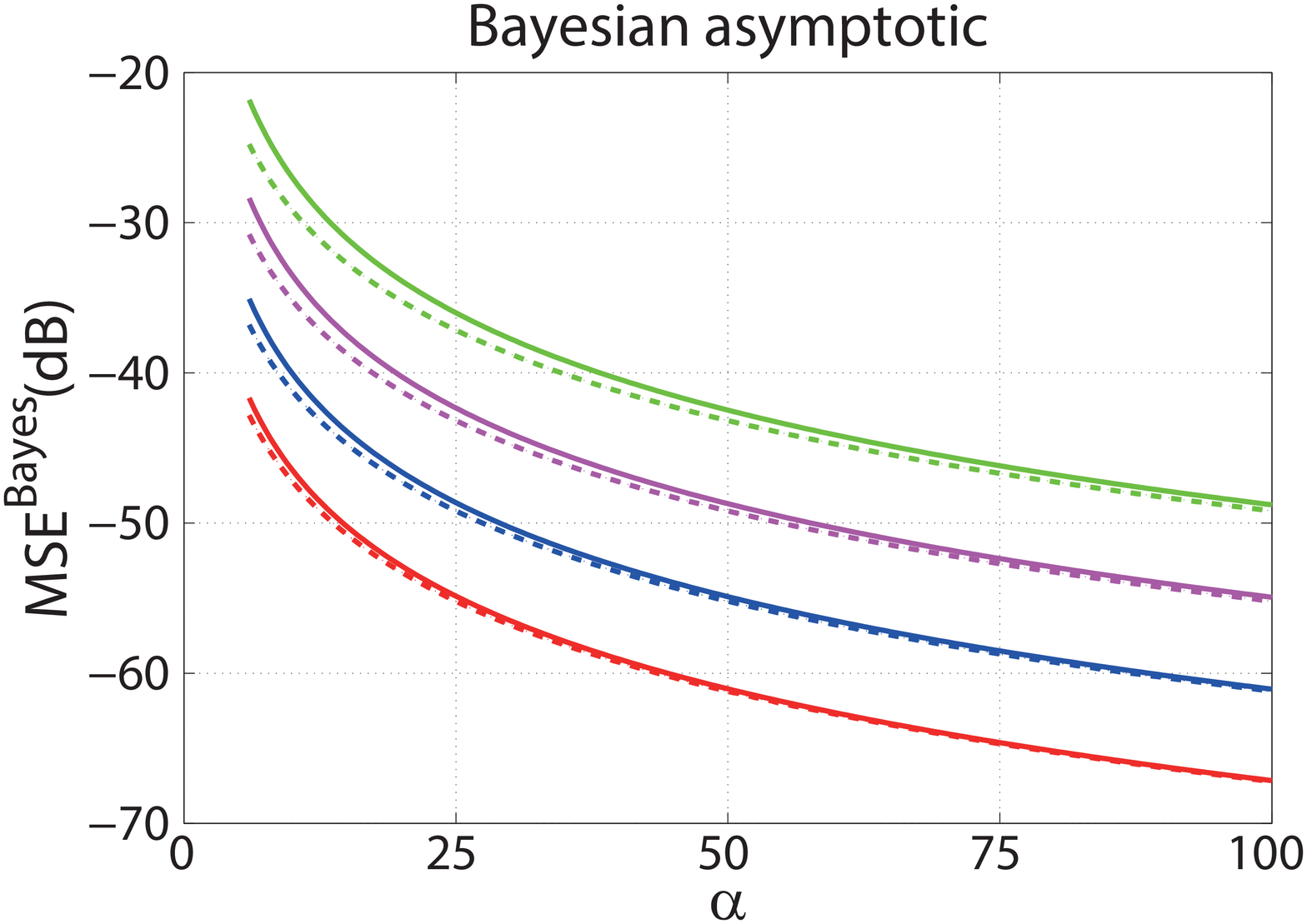}
          \hspace{1.6cm} 
        \end{center}
      \end{minipage}
      
      \begin{minipage}{0.5\hsize}
        \begin{center}
          \includegraphics[clip, angle=0,width=8cm]{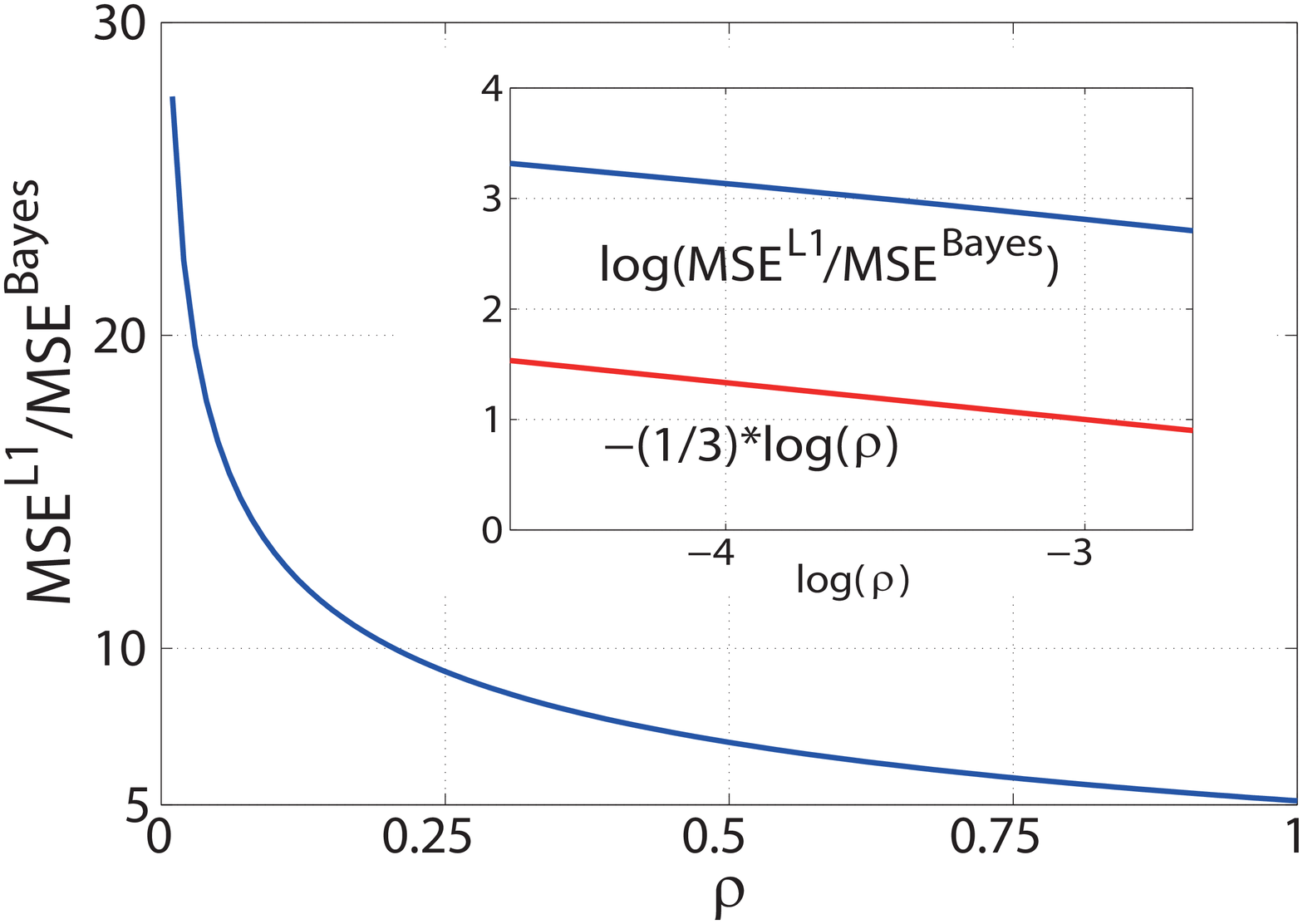}
          \hspace{1.6cm}
        \end{center}
      \end{minipage}
    \end{tabular}
\caption{Left: MSE (in decibels) versus measurement bit ratio $\alpha$ for Bayesian optimal signal reconstruction of 1-bit CS  {for Gauss-Bernoulli prior}.
Red, blue, magenta, and green correspond to {$\rho=0.03125, 0.0625, 0.125$, and
$0.25$, respectively. }
The solid curves represent the theoretical prediction obtained by (\ref{m_result}) and (\ref{mhat_result});
dashed curves show the performance when the positions of non-zero entries are known,
and dotted curves denote the asymptotic forms (\ref{asymptotics}), 
{which are indistinguishable from the dashed curves because they closely overlap. }
Right: Ratio of MSE between $l_{1}$-norm and Bayesian approaches when $\alpha \gg 1$ versus sparsity $\rho$ of the signal.
The inset shows a log-log plot for $0<\rho<0.1$. 
The least-squares fit implies that the ratio diverges as $O(\rho^{-0.33})$ as $\rho \to 0$. }
\label{fig.asymptotic}  
\end{center}
\end{figure}

Equation (\ref{asymptotics}) means that, at least in terms of MSEs,  correct prior knowledge of the sparsity {asymptotically} becomes 
as informative as the knowledge of the exact positions of the non-zero components. 
In most statistical models, the accuracy of asymptotic inference is expressed as a function of 
the ratio $\alpha=M/N$ between the number of data $M$ and the dimensionality of 
the variables to be inferred $N$ \cite{Sompolinsky1993,WatkinRauBiehl1993}. 
Equation (\ref{asymptotics}) indicates that, in the current problem, the dimensionality $N$ is replaced with 
the actual degree of the non-zero components $N\rho$, which originates from the 
singularity of the prior distribution (\ref{sparse}). 
This implies that caution is necessary in testing the validity of statistical models when 
sparse priors are employed, since conventional information criteria such as Akaike's information criterion 
\cite{Akaike1973} and the minimum description length \cite{RIssanen1978} mostly handle 
objective statistical models that are free of singularities, so that the model complexity 
is naively incorporated as the number of parameters $N$ \cite{Watanabe2009}. 

Equation (\ref{MSE_L1}) indicates that, even if prior knowledge of the sparsity is not available, 
optimal convergence can be achieved in terms of the ``exponent  {(decay of 
$O(\alpha^{-2})$})'' as $\alpha \to \infty$ using the 
$l_1$-norm approach. However, the performance can differ considerably in terms of the ``pre-factor  {(coefficient of $\alpha^{-2}$)}'' 
The right panel of Figure \ref{fig.asymptotic} plots the ratio $\mathrm{MSE}^{l_1}/\mathrm{MSE^{Bayes}}$, 
which diverges as $O(\rho^{-0.33})$ as $\rho \to 0$. This indicates that 
prior knowledge of the sparsity of the objective signal is more beneficial as $\rho$ becomes smaller.

\begin{figure}[t]
  \begin{center}
    \begin{tabular}{c}
    
      \begin{minipage}{0.5\hsize}
        \begin{center}
          \includegraphics[clip, angle=0,width=8cm]{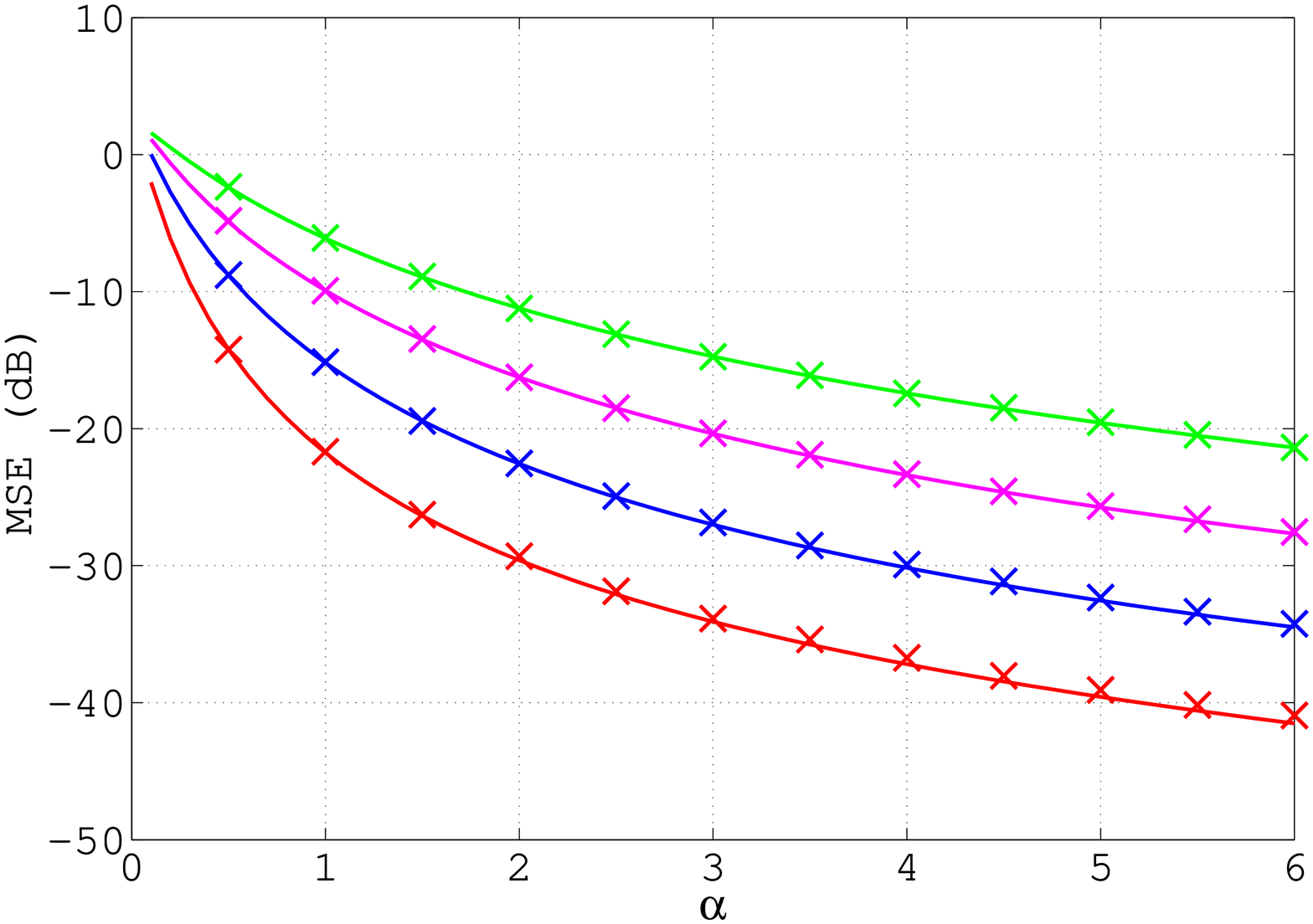}
          \hspace{1.6cm} 
        \end{center}
      \end{minipage}
      
      \begin{minipage}{0.5\hsize}
        \begin{center}
          \includegraphics[clip, angle=0,width=8cm]{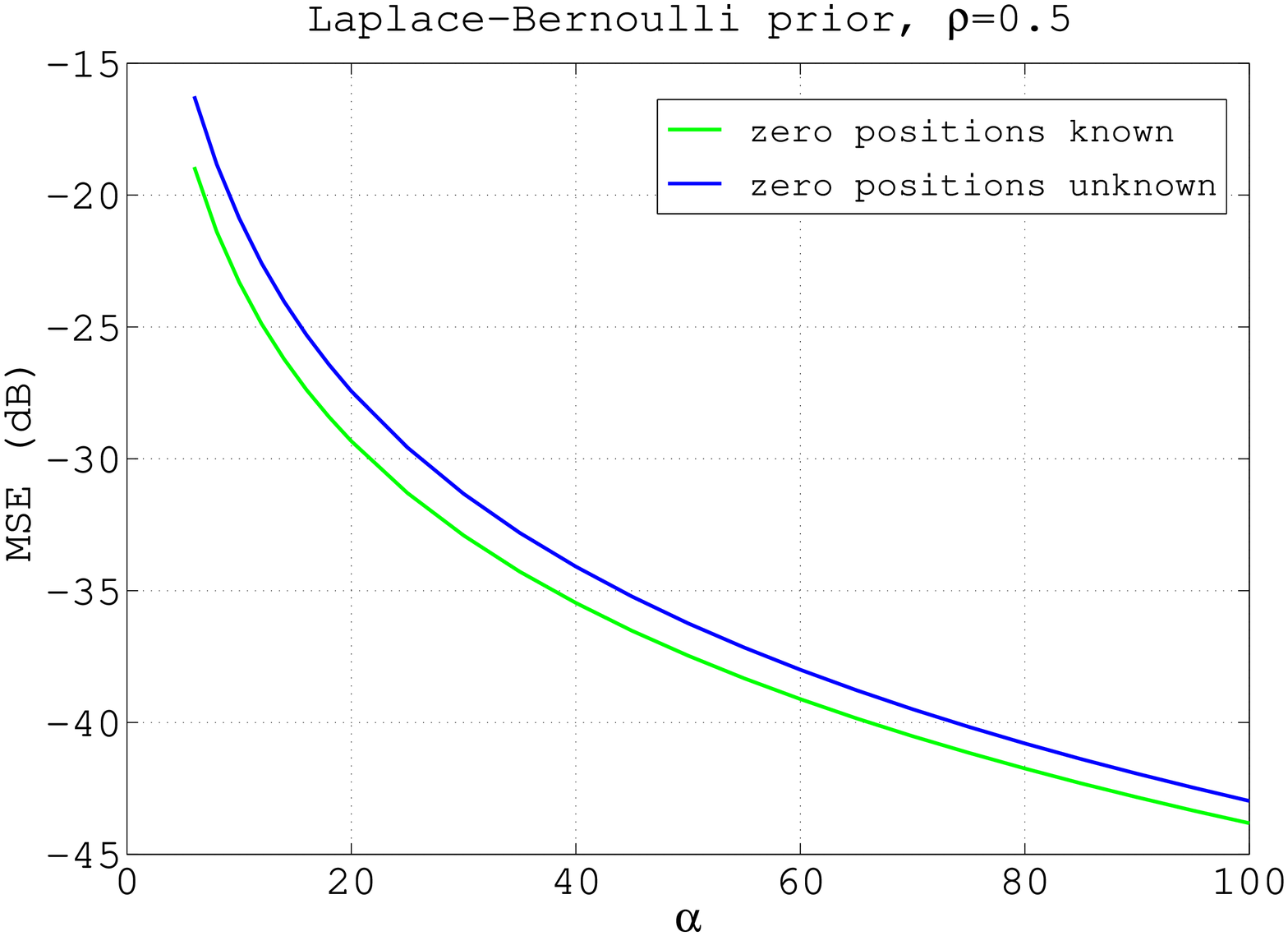}
          \hspace{1.6cm}
        \end{center}
      \end{minipage}
    \end{tabular}
\caption{
 {Left: MSE (in decibels) versus $\alpha$ for 1-bit CS in the case of Laplace-Bernoulli prior. 
Solid lines represent the theoretical prediction and the markers represent the experiment results by 1bitAMP algorithm for a signal size of $N=1024$ and averaged from 1000 experiments. 
Red, blue, magenta, and green represent $\rho=0.03125, 0.0625, 0.125$, and $0.25$, respectively.
Right: Asymptotic behavior of MSE for Laplace-Bernoulli prior, when the positions of zero entries of the signal are known and unknown. This implies that MSE of these two cases are different even asymptotically.
 }}
\label{Laplace_MSEdb}
\end{center}
\end{figure}

 {
For checking the generality of the results obtained for Gauss-Bernoulli prior, 
we also carried out similar analysis for Laplace-Bernoulli prior
\begin{equation}
P\left(x\right) =\left(1-\rho\right) \delta \left( x \right)+\frac{\rho}{2}e^{-|x|}.
\end{equation}
The left panel of Fig.~\ref{Laplace_MSEdb} shows the comparison between the replica prediction and 
the experimental results by GAMP, which supports that the replica and GAMP correspondence 
does hold for general priors. 
The right panel of Fig.~\ref{Laplace_MSEdb} compares the performance with that achieved when the positions 
of non-zero entries are known. Unlike the case of Gauss-Bernoulli prior, the two performances
do not get close even asymptotically. This implies that the significance of utility of the Bayesian approach 
depends considerably on the statistical property of the objective signal. 
}

\section{Summary}
In summary, we have examined the typical performance of {the Bayesian optimal signal recovery for 1-bit CS} using methods from statistical mechanics.
 {For Gauss-Bernoulli prior,} 
using the replica method to compare the performance of the Bayesian optimal approach to the $l_1$-norm minimization,  
we have shown that the utility of correct prior knowledge on the objective signal, 
which is incorporated in the Bayesian optimal scheme, becomes more significant 
as the density of non-zero entries $\rho$ in the signal decreases. 
In addition, we have clarified that,  {for this particular prior,} 
the MSE performance asymptotically saturates that obtained when the exact positions 
of non-zero entries are exactly known as the number of 1-bit measurements increases. 
We have also developed a practically feasible approximate algorithm for Bayesian signal recovery, 
which can be regarded as a special case of the GAMP algorithm.
The algorithm has a computational cost of the square of the system size per update, 
exhibiting a fairly good convergence property as the system size 
becomes larger. The experimental results  {for both Gauss-Bernoulli prior and Laplace-Bernoulli prior} show excellent agreement with the predictions made by the replica method.
These indicate that almost-optimal reconstruction performance can be 
attained with a computational complexity of the square of the signal length per update
 {for general priors}, 
which is highly beneficial in practice. 

Obtaining the correct prior distribution of the sparse signal may be an obstacle 
to applying the current approach in practical problems. One possible solution is to estimate 
hyper-parameters that characterize the prior distribution in the reconstruction stage, 
as has been proposed for normal CS \cite{Krzakala2012}. 
It was reported that orthogonal measurement matrices, rather than those of statistically independent entries,
enhance the signal reconstruction performance for several problems related to  
CS \cite{ShizatoKabashima2009,KabaVehkaperaChatterjee2013,VehkaperaKabaChatterjee2014,KabaVehkapera2014,
OymakHassibi2014,WenWong2014}.
Such devices may also be effective for 1-bit CS.

\ack
YX is supported by JSPS Research Fellowships DC2.  
This study was also partially supported by the JSPS Core-to-Core Program ``Non-equilibrium dynamics of soft matter and information,'' 
JSPS KAKENHI Nos. 26011287 (YX), 25120013 (YK), 
{and} {the Grant DySpaN of Triangle de la Physique (LZ).
Useful discussions with Chistophe Sch\"{u}lke are also acknowledged. }
\appendix

\section{Derivation of $(\ref{eq:free energy})$}
\label{replicaderivation}
\subsection{Assessment of $\left [P^n \left(\vm{y} |\vm{\Phi} \right) \right]_{\vm{\Phi},\vm{y}}$ for $n \in \mathbb{N}$}
Averaging (\ref{eq:expansion}) with respect to $\vm{\Phi}$ and $\vm{y}$ gives the following expression for the $n$-th moment of the partition function:
\begin{eqnarray}
\left [P^n \left(\vm{y} |\vm{\Phi} \right) \right]_{\vm{\Phi},\vm{y}}
=\int \prod_{a=1}^n \left (d \textrm{\boldmath $x$}^a P\left( \textrm{\boldmath $x^a$}\right)\right )
 \times\left[\prod_{a=1}^n \prod_{\mu=1}^M \Theta\left ( (\vm{y})_\mu  (\vm{\Phi} \vm{x}^a)_\mu \right )\right]_{\vm{\Phi},\vm{y}}. 
\label{Zmoment}
\end{eqnarray}
We insert $n(n+1)/2$ trivial identities
\begin{eqnarray}
1=N \int dq_{ab} \delta \left (\vm{x}^a \cdot \vm{x}^b - N q_{ab} \right ), 
\end{eqnarray}
where $a>b=0,1,2,\ldots,n$, into (\ref{Zmoment}). Furthermore, we define a joint distribution of $n+1$ vectors $\{\vm{x}^a\}=\{\vm{x}^0, \vm{x}^1,\vm{x}^2,\ldots,\vm{x}^n \}$ as
\begin{eqnarray}
&&P\left (\{\vm{x}^a\} |\vm{Q}\right )= \frac{1}{V\left (\vm{Q}\right ) } P(\vm{x}^0) \times 
\prod_{a=1}^n \left (P\left( \textrm{\boldmath $x^a$}\right)\right ) 
 \times \prod_{a>b} \delta \left (\vm{x}^a \cdot \vm{x}^b - N q_{ab} \right ), 
\label{replica_joint_dist}
\end{eqnarray}
where $\vm{Q}=(q_{ab})$ is an $(n+1) \times (n+1)$ symmetric matrix 
whose $00$ and other diagonal entries are fixed as $\rho$ and $Q$, respectively.
 $P(\vm{x}^0)=\prod_{i=1}^N \left ((1-\rho)\delta(x_i^0)+\rho \tilde{P}(x_i^0) \right )$ denotes the distribution of the original signal $\vm{x}^0$, and $V\left (\vm{Q}\right )$ is the normalization constant that ensures $\int \prod_{a=0}^n d \vm{x}^a P\left (\{\vm{x}^a\} |\vm{Q}\right )=1$ holds. These indicate that (\ref{Zmoment}) can also be expressed as
\begin{eqnarray}
\left [P^n \left(\vm{y} |\vm{\Phi} \right) \right]_{\vm{\Phi},\vm{y}}
=\int d\vm{Q} \left (V\left (\vm{Q}\right ) \times \Xi\left (\vm{Q}\right ) \right ), 
\label{Zn}
\end{eqnarray}
where $d\vm{Q} \equiv \prod_{a>b}d q_{ab}$ and 
\begin{eqnarray}
\Xi\left (\vm{Q}\right )
=\int \prod_{a=0}^n d\vm{x}^a P\left (\{\vm{x}^a\} |\vm{Q} \right ) 
\left[ {\displaystyle \sum_{\vm{y}}} {\prod_{a=0}^n }\prod_{\mu=1}^M 
\Theta\left(
((\vm{y})_\mu 
(\vm{\Phi} \vm{x}^a)_\mu 
\right)
\right]_{\vm{\Phi}}. 
\label{Xi}
\end{eqnarray}

Equation (\ref{Xi}) can be regarded as the average of $ {\displaystyle \sum_{\vm{y}}\prod_{a=0}^n }\prod_{\mu=1}^M \Theta\left((\vm{y})_\mu (\vm{\Phi} \vm{x}^a)_\mu \right)$ with respect to $\{\vm{x}^a\}$ and $\vm{\Phi}$ over distributions of $P\left (\{\vm{x}^a\} \right )$ and $P(\vm{\Phi}) \equiv \left (\sqrt{2\pi/N} \right )^{-MN} \exp \left (-(N/2) \sum_{\mu,i} \Phi_{\mu i}^2 \right )$. 
In computing this,  note that the central limit theorem guarantees that $u_\mu^a \equiv (\vm{\Phi} \vm{x}^a)_\mu = \sum_{i=1}^N \Phi_{\mu i} x_i^a$ can be handled as zero-mean multivariate Gaussian random numbers whose variance and covariance are given by 
\begin{eqnarray}
\left [u_\mu^a u_\nu^b \right ]_{\vm{\Phi},\{\vm{x}^a\}}
=\delta_{\mu \nu} q_{ab}, 
\end{eqnarray}
when $\vm{\Phi}$ and $\{\vm{x}^a\}$ are generated independently from $P(\vm{\Phi})$ and $P\left (\{\vm{x}^a\} \right )$, respectively. This means that (\ref{Xi}) can be evaluated as
\begin{eqnarray}
\Xi(\vm{Q}) &=&\left (\frac{\int d \vm{u} \exp\left (-\frac{1}{2}\vm{u}^{\rm T} \vm{Q}^{-1} \vm{u} \right )
 {
\displaystyle \sum_{y \in \{+1,-1\}}}\prod_{a=0}^n 
\Theta\left (y u^a \right ) }{
(2\pi)^{(n+1)/2} (\det \vm{Q})^{1/2}} 
\right )^M \cr
&=& \left (2 \int \frac{ d \vm{u}\exp\left (-\frac{1}{2}\vm{u}^{\rm T} \vm{Q}^{-1} \vm{u} \right )
\prod_{a=0}^n \Theta\left (u^a \right )  }{
(2\pi)^{(n+1)/2} (\det \vm{Q})^{1/2}}
 \right )^M. 
\label{logXi}
\end{eqnarray}

On the other hand, expressions 
\begin{eqnarray}
\delta\left (|\vm{x}^a|^2-N Q\right )
=\frac{1}{4 \pi} \int_{-{\rm i}\infty}^{+{\rm i}\infty}
d \hat{q}_{aa} \exp \left (-\frac{1}{2} \hat{q}_{aa} 
\left (|\vm{x}^a|^2-N Q\right ) \right )
\end{eqnarray}
and 
\begin{eqnarray}
\delta\left (\vm{x}^a \cdot \vm{x}^b -N q_{ab}\right )
=\frac{1}{2 \pi} \int_{-{\rm i}\infty}^{+{\rm i}\infty}
d \hat{q}_{ab} \exp \left ( \hat{q}_{ab} 
\left (\vm{x}^a \cdot \vm{x}^b-N q_{ab} \right ) \right ), 
\end{eqnarray}
and use of the saddle-point method, offer 
\begin{eqnarray}
&&\frac{1}{N} \log V(\vm{Q})
=\mathop{\rm extr}_{\hat{\vm{Q}}}
\left \{
-\frac{1}{2} {\rm Tr} \hat{\vm{Q}}\vm{Q} \right . \cr
&& \hspace*{2cm}
\left . + \log 
\left (
\int d\vm{x} 
P(x^0) \prod_{a=1}^n P(x^a)\exp \left (\frac{1}{2} 
\vm{x}^{\rm T} \hat{\vm{Q} }\vm{x} \right )
\right ) 
\right \}.  
\label{logV}
\end{eqnarray}
Here, $\vm{x}=(x^0,x^1,\ldots,x^n)^{\rm T}$ and $\hat{\vm{Q}}$ is an $(n+1)\times (n+1)$ symmetric matrix whose $00$ and other diagonal components are given as $0$ and $-\hat{q}_{aa}$, respectively. The off-diagonal entries are  $\hat{q}_{ab}$. 
Equations (\ref{logXi}) and (\ref{logV}) indicate that $N^{-1} \log \left [P^n \left(\vm{y} |\vm{\Phi} \right) \right]_{\vm{\Phi},\vm{y}}$ is correctly evaluated by  the saddle-point method with respect to $\vm{Q}$ in the assessment of the right-hand side of (\ref{Zn}), when $N$ and $M$ tend to infinity and $\alpha=M/N$ remains finite. 

\subsection{Treatment under the replica symmetric ansatz}
Let us assume that the relevant saddle-point for assessing (\ref{Zn}) is of the form (\ref{RSanzats}) and, accordingly, 
\begin{eqnarray}
\hat{q}_{ab}=\hat{q}_{ba} =\left \{
\begin{array}{ll}
0, &( \mbox{$a=b=0$})\\
\hat{m}, &( \mbox{$a=1,2,\ldots,n$; $b=0$})\\
\hat{Q}, &( \mbox{$a=b=1,2,\ldots,n$})\\
\hat{q}, &( \mbox{$a\ne b =1,2,\ldots,n$})
\end{array}
\right .  .
\label{RShatQ}
\end{eqnarray}
The $n+1$-dimensional Gaussian random variables $u^0,u^1,\ldots u^n$ whose variance and covariance are given by (\ref{RSanzats}) can be expressed as 
\begin{eqnarray}
&& u^0=\sqrt{\rho-\frac{m^2}{q}}s^0 + \frac{m}{\sqrt{q}} z, \label{Newgauss0}\\
&& u^a=\sqrt{Q-q} s^a+ \sqrt{q} z, \ (a=1,2,\ldots,n) \label{Newgaussa} 
\end{eqnarray}
utilizing $n+2$ independent standard Gaussian random variables $z$ and $s^0,s^1,\ldots,s^n$. This indicates that (\ref{logXi}) is evaluated as
\begin{eqnarray}
\Xi(\vm{Q})=\left (
2 \int {\rm D}z 
H\left (\frac{m}{\sqrt{\rho q-m^2} }z \right )
H^n \left (\sqrt{\frac{q}{Q-q}} z \right ) \right )^M. 
\label{NewXi}
\end{eqnarray}
On the other hand, substituting (\ref{RShatQ}) into (\ref{logV}), in conjunction with the identity
\begin{eqnarray}
\exp \left (\hat{q} \sum_{a>b(\ge 1)} x^a x^b \right )=\int {\rm D}z \exp \left (\sum_{a=1}^n 
\left (
-\frac{\hat{q}}{2} (x^a)^2 + \sqrt{\hat{q}} z x^a \right ) \right ), 
\end{eqnarray}
provides 
\begin{eqnarray}
&&\frac{1}{N} \log V(\vm{Q} )=\mathop{\rm extr}_{\hat{Q}, \hat{q},\hat{m} }
\left \{
\frac{n}{2}\hat{Q}Q-\frac{n(n-1)}{2} \hat{q}q -\hat{m} m \right . \cr
&&\left . 
+\log 
\left [\left (
\!
\int 
\!
dx P(x) \exp 
\!\left (
\!
-\frac{\hat{Q}
\!+\! \hat{q}}{2} x^2\!+\!\left (\!\sqrt{\hat{q}}z \!+\! \hat{m} x^0 \!\right )
\!x \!\right )
\right )^n \right ]_{x^0,z} \right \}. 
\label{NewV}
\end{eqnarray}
Although we have assumed that $n \in \mathbb{N}$, the expressions of (\ref{NewXi}) and (\ref{NewV}) are likely to hold for $n \in \mathbb{R}$ as well. 
Therefore, the average free energy $\overline{f}$ can be evaluated by substituting these expressions into the formula $\overline{f}= \lim_{n\to 0} (\partial/\partial n) \left ((N)^{-1} \log \left [P^n \left(\vm{y} |\vm{\Phi} \right) \right]_{\vm{\Phi},\vm{y}} \right )$. 

Furthermore,  {employing the 
expressions that hold for $|n| \ll 1$, $
H^{n}(x)=
\textrm{exp}\left( n \log H(x) \right)$ $\approx 1+n\log H(x)$ and $
\log\left(1+nC(\cdot)\right)\approx nC(\cdot)$, where $C(\cdot)$ is an arbitrary function,}
 we obtain the form
\begin{eqnarray}
\lim_{n\to 0} \frac{\partial}{\partial n} \frac{1}{N} \log 
\Xi(\vm{Q})=
2 \alpha\int {\rm D}z 
H\!\left (\!\frac{m}{\sqrt{\rho q-m^2} }z \!\right )\!
 {\log }H\!\left (\!\sqrt{\frac{q}{Q-q}} z \!\right )\!. 
\end{eqnarray}
And we have 
\begin{eqnarray}
\lim_{n\to 0} \frac{\partial}{\partial n} \frac{1}{N} \log
V(\vm{Q} )
&=&\mathop{\rm extr}_{\hat{Q},\hat{q},\hat{m}} \!\Biggr\{ \int \textrm{d}x^0 P\left( x^0\right) \int\textrm{D}z\phi \left(\sqrt{\hat{q}}z+\hat{m}x^{0};\hat{Q}\right)\nonumber\\
&&+\frac{1}{2}Q\hat{Q}+\frac{1}{2}q\hat{q}-m\hat{m} \Biggl\} . 
\end{eqnarray}
Using these in the resultant expression of $\overline{f}$ gives (\ref{eq:free energy}). 

\section{Derivation of (\ref{m_mu_i_3})--(\ref{g_out_p})}
\label{GAMP_derivation}
Expanding the exponential in (\ref{m_mu_i_2}) up to the second order of $\Phi_{\mu i}(x_i -a_{i\to \mu})$ and 
performing the integration with respect to  $u_\mu$ gives
\begin{eqnarray}
m_{\mu\rightarrow i}\left(x_i\right)
&\simeq& c_{0}+c_{1}\Phi_{\mu i}\left(x_i - a_{i\rightarrow\mu}\right)
 +\frac{1}{2}c_{2}\Phi_{\mu i}^2\left(x_i -a_{i\rightarrow \mu}\right)^2\nonumber\\
&\simeq& \exp{\Biggl\{ \textrm{ln}c_0 +\frac{c_1}{c_0}\Phi_{\mu i} \left(x_i - a_{i\rightarrow\mu}\right) 
  +\frac{c_{0}c_{2}-c_{1}^2}{2c_{0}^2}\Phi_{\mu i}^2\left(x_i -a_{i\rightarrow \mu}\right)^2\Biggr\}}\nonumber\\
&\propto& \exp{\Biggl\{ -\frac{A_{\mu\rightarrow i}}{2}x_{i}^2 +B_{\mu\rightarrow i}x_i \Biggr\}},
\label{m_mu_i_3_appendix}
\end{eqnarray}
where
\begin{eqnarray}
c_{0}\equiv\int\textrm{d}u_{\mu} P(y_{\mu}|u_{\mu})
\textrm{exp}\left(-\frac{(u_{\mu}-\omega_{\mu})^2}{2V}\right) ,\label{c0}\\
c_{1}\equiv\int\textrm{d}u_{\mu} P(y_{\mu}|u_{\mu})\left(\frac{u_{\mu}-\omega_{\mu}}{V}\right)
\textrm{exp}\left(-\frac{(u_{\mu}-\omega_{\mu})^2}{2V}\right), \label{c1}\\
c_{2}\equiv\int\textrm{d}u_{\mu} P(y_{\mu}|u_{\mu})
\left(\left(\frac{u_{\mu}-\omega_{\mu}}{V}\right)^2-\frac{1}{V}\right)
\textrm{exp}\left(-\frac{(u_{\mu}-\omega_{\mu})^2}{2V}\right),\label{c2}
\end{eqnarray}
and 
\begin{eqnarray}
A_{\mu\rightarrow i}=\frac{c_1^2-c_0c_2}{c_0^2}\Phi_{\mu i}^2,\label{A_appendix}\\
B_{\mu\rightarrow i}=\frac{c_1}{c_0}\Phi_{\mu i}+\frac{c_1^2-c_0c_2}{c_0^2}\Phi_{\mu i}^2 a_{i\rightarrow \mu}.\label{B_appendix}
\end{eqnarray}
Equations (\ref{c1}) and (\ref{c2}) imply that $c_1$ and $c_2$ can be expressed as $c_1=\partial c_0/\partial \omega_\mu$ and $c_2=\partial^2c_0/\partial \omega_\mu^2$, 
respectively. Inserting this into (\ref{A_appendix}) and (\ref{B_appendix}), we obtain (\ref{m_mu_i_3})--(\ref{g_out_p}).

\section{Asymptotic form of $\mathrm{MSE^{Bayes}}$}
\label{app3}
The behavior as $m\to \rho$ and $\hat{m} \to \infty$ is obtained as $\alpha \to \infty$.  
This implies that,  {for Gauss-Bernoulli distribution}, {equations}~(\ref{m_result}) and (\ref{mhat_result}) can be evaluated as
\begin{eqnarray}
m &=& \int {\rm D}t \frac{\rho^2 (1+\hat{m})^{-1} e^{\frac{\hat{m}}{1+\hat{m}}t^2} \frac{\hat{m}}{(1+\hat{m})^2} t^2 }{1-\rho + 
\rho (1+\hat{m})^{-1/2} e^{\frac{\hat{m}}{2(1+\hat{m})}t^2}} \cr
&=&  \frac{\rho^2\hat{m}}{(1+\hat{m})}
\int {\rm D} z z^2 \left [ (1-\rho) (1+\hat{m})^{1/2} e^{-\frac{\hat{m}}{2} z^2}+\rho \right ]^{-1} \cr
&\simeq& \rho (1-\hat{m}^{-1})
\label{app3m}
\end{eqnarray}
and 
\begin{eqnarray}
\hat{m} &=& \frac{2 \alpha}{\rho-m} \int {\rm D} t \frac{e^{-\frac{m}{\rho-m}t^2}/(2\pi)}{H\left (\sqrt{\frac{m}{\rho-m}} t \right ) } 
=\frac{2\alpha}{\sqrt{m(\rho-m)}} \int \frac{{\rm d}z}{(2 \pi)^{3/2}} \frac{e^{-\frac{\rho+m}{2m} z^2}}{H(z)} \cr
&\simeq & \frac{2 C \alpha }{\sqrt{m(\rho-m)}}, 
\label{app3mhat}
\end{eqnarray}
respectively. Here,  the integration variables have been changed to $(1+\hat{m})^{-1/2}t = z$ and $\sqrt{m/(\rho-m)} t = z$ 
in (\ref{app3m}) and (\ref{app3mhat}), respectively, and we set $C \equiv \int {\rm d}z (2 \pi)^{-3/2} e^{-z^2}/H(z) = 0.3603\ldots$. 
Equations (\ref{app3m}) and (\ref{app3mhat}) yield an asymptotic expression for $m$:  
\begin{eqnarray}
m\simeq  \rho\left (1-\left (\frac{\rho}{2 C\alpha} \right )^2 \right ). 
\end{eqnarray}
Inserting this into (\ref{MSE}) gives (\ref{asymptotics}). 

The performance when the positions of non-zero entries are known can be evaluated by setting $\rho =1$ and replacing 
$\alpha $ with $\alpha/\rho$ in (\ref{m_result}) and (\ref{mhat_result}) as the dimensionality of $\vm{x}$ is reduced from 
$N$ to $N\rho$. This reproduces (\ref{asymptotics}) in the asymptotic region
of $\alpha \gg 1$.

\section{Asymptotic form of $\mathrm{MSE}^{l_1}$}
\label{app4}
The saddle-point equations of the $l_1$-norm minimization approach 
{under a normalization constraint of $|\vm{x}|^2=N$}
are as follows \cite{YingKaba1bitCS2013}:
\begin{eqnarray}
\hat{q}\!&=&\! \frac{\alpha}{\pi \chi^2} \!\left(\!\arctan \!\left(\!\frac{\sqrt{\rho\!-\!m^2}}{m}\right)\!-\!\frac{m}{\rho}\sqrt{\rho\!-\!m^2}\right), \label{qh_l1}\\
\hat{m}&=&\frac{\alpha}{\pi\chi\rho}\sqrt{\rho-m^2}\label{mh_l1}, \\
\hat{Q}^{2}\!&=&\! 2\left\{ \left(1\!-\!\rho\right)\left[\left(\hat{q}+1\right)H\left(\frac{1}{\sqrt{\hat{q}}}\right)-\sqrt{\frac{\hat{q}}{2\pi}}e^{-\frac{1}{2\hat{q}}} \right] \right. \nonumber\\
          && +\rho\left[\left(\hat{q}+\hat{m}^2+1\right)H\left(\frac{1}{\sqrt{\hat{q}+\hat{m}^2}}\right) \right . \nonumber\\
          && \left . \left . -\sqrt{\frac{\hat{q}+\hat{m}^2}{2\pi}}e^{-\frac{1}{2\left(\hat{q}+\hat{m}^2\right)}}\right]\right\}\label{Qh2_l1}, \\
\chi\!&=&\! \frac{2}{\hat{Q}}\!\left[\!\left(\!1\!-\!\rho\right)H\left(\frac{1}{\sqrt{\hat{q}}}\right) 
    \!+\!\rho H\!\left(\!\frac{1}{\sqrt{\hat{q}\!+\!\hat{m}^2}}\!\right)\!\right]\label{chi_l1}, \\
m\!&=&\!\frac{2\rho\hat{m}}{\hat{Q}}H\left(\frac{1}{\sqrt{\hat{q}+\hat{m}^2}}\right)\label{m_l1}. 
\end{eqnarray}

The behavior as $m\to \sqrt{\rho}$ and $\hat{m} \to \infty$ is obtained as $\alpha \to \infty$.  
This implies that~(\ref{Qh2_l1})  can be evaluated as
 \begin{eqnarray}
\hat{Q} \!&\simeq&\! \left( \rho\hat{m}^2\!-\!\frac{4\hat{m}}{\sqrt{2\pi}}\!+\!
B(\hat{q},\rho) \right)^{1/2}\nonumber\\
\!&\simeq&\! \sqrt{\rho}\hat{m}\left[ 1-\frac{2}{\sqrt{2\pi}\hat{m}}+\left (\frac{B(\hat{q},\rho)}{2\rho}-\frac{3}{\pi} \right )\right],
\label{Qh_l1_asym}
\end{eqnarray}
where $B(\hat{q},\rho)\equiv \rho\left(\hat{q}\!+\!1\right)\!+\!
2\left(1\!-\!\rho\right)\left[\left(\hat{q}+1\right)H\left(\frac{1}{\sqrt{\hat{q}}}\right)-\sqrt{\frac{\hat{q}}{2\pi}}e^{-\frac{1}{2\hat{q}}} \right]$. 
Inserting~(\ref{Qh_l1_asym}) into~(\ref{m_l1}), we obtain
\begin{eqnarray}
m\simeq\sqrt{\rho}\left(1-\delta\right), 
\label{m_l1_asym}
\end{eqnarray}
where 
\begin{eqnarray}
\delta\equiv \left (\frac{B(\hat{q},\rho)}{2\rho}-\frac{1}{\pi} \right )/\hat{m}^2=\pi^2\left[2(1-\rho)H \left (1/\sqrt{\hat{q}} \right )+\rho\right]^2/(2\alpha^2).
\label{delta_l1}
\end{eqnarray}
Inserting (\ref{mh_l1}), (\ref{Qh_l1_asym}), (\ref{m_l1_asym}), and
$\chi\simeq\left[2(1-\rho)H(1/\sqrt{\hat{q}})\right]/\hat{Q}$ into 
(\ref{qh_l1}) yields a closed equation with respect to $\hat{q}$: 
\begin{eqnarray}
\hat{q}\simeq\frac{2}{3}\left(B(\hat{q},\rho) -\frac{2\rho}{\pi}\right)\left[2(1-\rho)H(1/\sqrt{\hat{q}})+\rho\right]^{-1}. 
\label{qh_l1_asym}
\end{eqnarray}
This determines the value of $\hat{q}$ for $\alpha \to \infty$, $\hat{q}_{l_1}^\infty(\rho)$. 
Combining (\ref{delta_l1}) and 
\begin{eqnarray}
{\rm MSE}^{l_1}
=2\left (1-\frac{m}{\sqrt{\rho}} \right )
\simeq2\delta
\end{eqnarray}
gives (\ref{MSE_L1}) in the asymptotic region of $\alpha\gg1$.


\end{document}